\documentclass[draftclsnofoot,onecolumn,12pt]{IEEEtran}
\usepackage{amssymb}
\usepackage[cmex10]{amsmath}
\usepackage{stfloats}
\usepackage{graphicx}
\usepackage{multirow}
\usepackage{subfigure}
\usepackage{tabularx}
\usepackage{booktabs}
\usepackage{amsmath}
\usepackage{epsfig,epsf,color,balance,cite}
\usepackage{indentfirst}
\usepackage{mathrsfs}
\usepackage[noend]{algpseudocode}
\usepackage{algorithmicx,algorithm}

\newtheorem{theorem}{Theorem}
\newtheorem{lemma}{Lemma}
\newtheorem{remark}{Remark}
\newtheorem{corollary}{Corollary}

\begin{document}
\title{On the Secrecy Rate of Spatial Modulation Based Indoor Visible Light Communications}

\author{Jin-Yuan Wang, Hong Ge, Min Lin, Jun-Bo Wang, Jianxin Dai, \\and Mohamed-Slim Alouini, \emph{Fellow, IEEE}
\thanks{Jin-Yuan Wang Hong Ge and Min Lin are with College of Telecommunications and Information Engineering, Nanjing University of Posts and Telecommunications, Nanjing 210003, China. (E-mail: jywang@njupt.edu.cn, 1217012002@njupt.edu.cn, linmin@njupt.edu.cn)}
\thanks{Jun-Bo Wang is with National Mobile Communications Research Laboratory, Southeast University, Nanjing 210096, China. (E-mail: jbwang@seu.edu.cn)}
\thanks{Jianxin Dai is with College of Science, Nanjing University of Posts and Telecommunications, Nanjing 210003, China. (E-mail: daijx@njupt.edu.cn)}
\thanks{Mohamed-Slim Alouini is with Computer, Electrical and Mathematical Science and Engineering Division, King Abdullah University of Science and Technology, Thuwal 23955-6900, Saudi Arabia. (E-mail: slim.alouini@kaust.edu.sa)}
\thanks{Corresponding author: Jin-Yuan Wang (jywang@njupt.edu.cn)}
}

\maketitle
\linespread{1.5}
\begin{abstract}
In this paper, 
we investigate the physical-layer security for a spatial modulation (SM) based indoor visible light communication (VLC) system, 
which includes multiple transmitters, a legitimate receiver, and a passive eavesdropper (Eve). 
At the transmitters, the SM scheme is employed, i.e., only one transmitter is active at each time instant. 
To choose the active transmitter, a uniform selection (US) scheme is utilized. 
Two scenarios are considered: one is with non-negativity and average optical intensity constraints, 
the other is with non-negativity, average optical intensity and peak optical intensity constraints.
Then, lower and upper bounds on the secrecy rate are derived for these two scenarios.
Besides, the asymptotic behaviors for the derived secrecy rate bounds at high signal-to-noise ratio (SNR) are analyzed.
To further improve the secrecy performance, a channel adaptive selection (CAS) scheme and a greedy selection (GS) scheme are proposed to select the active transmitter.
Numerical results show that the lower and upper bounds of the secrecy rate are tight. 
At high SNR, small asymptotic performance gaps exist between the derived lower and upper bounds. 
Moreover, the proposed GS scheme has the best performance, followed by the CAS scheme and the US scheme.
\end{abstract}

\begin{keywords}
Visible light communications,
Spatial modulation,
Physical-layer security,
Secrecy rate,
Asymptotic performance analysis,
Transmitter selection scheme.
\end{keywords}

\IEEEpeerreviewmaketitle

\newpage
\baselineskip=8.5mm
\section{Introduction}
\label{sec1}
For the fifth generation (5G) wireless communications,
multi-input multi-output (MIMO) will be employed as one of the promising technologies \cite{BIB00}.
However, by using multiple radio frequency (RF) chains,
the hardware cost of MIMO systems is very high.
To break such a limitation, spatial modulation (SM),
which employs only one RF chain, has been proposed as a low complexity solution \cite{BIB01,BIB01_1}.

The concept of SM in conventional RF wireless communications was first proposed by Chau and Yu \cite{BIB02}.
For SM systems with a multi-antenna transmitter,
only one antenna is activated at each time slot,
while other antennas remain silent.
A portion of the source data bits contains the index of the active antenna.
Therefore, the dimension of information is increased,
which can help to enhance the system performance.
The transceiver designs of SM were introduced in \cite{BIB03} and \cite{BIB04}.
At the receiver of SM systems, the active antenna index and the received signal should be estimated simultaneously.
To perform detection, the maximum likelihood detection \cite{BIB05},
matched filter based detection \cite{BIB06},
sphere decoding algorithm based detection \cite{BIB07},
and hybrid detection \cite{BIB08} were proposed.
Based on the transceiver design,
the performance indicators, such as channel capacity \cite{BIB09},
bit error rate \cite{BIB10}, and average bit error probability \cite{BIB11}, were investigated.
For a comprehensive introduction about SM, the readers can refer to \cite{BIB12}.

The large amount of research on SM has verified the advantages of SM over MIMO.
Recently, the investigation of SM has been extended to the field of visible light communications (VLC) \cite{BIB01}.
VLC is a novel data communication variant which uses visible light between 400-800 THz.
The concept of optical SM was proposed in \cite{BIB12_1},
while the SM applied to indoor VLC was discussed in \cite{BIB12_2}.
In indoor environment,
the SM was compared with the repetition coding and the spatial multiplexing in \cite{BIB12_3}.
Moreover, the constellation optimization design and the mutual information analysis for SM based VLC were investigated in \cite{BIB12_4} and \cite{BIB12_5}, respectively.
By using the channel state information (CSI), the channel adaptive SM schemes were analyzed in \cite{BIC1} and \cite{BIC3}.
To break the limitation that the number of required transmitters must be a power of two,
a channel adaptive bit mapping scheme was proposed in \cite{BIC2} for SM based VLC.
In \cite{BIC4}, a collaborative constellation based generalized SM encoding was presented.
In \cite{BIC5}, the impact of synchronization error on optical SM was analyzed.
In \cite{BIC6}, an iterative combinatorial symbol design algorithm was proposed for generalized SM in VLC.
Note that the above literatures do not consider the secure transmissions in the viewpoint of information-theoretic security.

In indoor VLC, the information security is a critical issue for users.
Owing to the line-of-sight propagation,
VLC is more secure than conventional RF wireless communications.
However, any receiver in VLC can receive information as long as it is located in the illuminated zone of the light-emitting diode (LED).
Therefore, such a feature still provides a possibility for unintended users to eavesdrop information.
To ensure secure transmission,
physical layer security techniques for indoor VLC have been proposed recently.
As it is known, the secrecy performance in VLC depends on the input distribution.
By employing the uniform \cite{BIC7},
truncated generalized normal \cite{BIC8},
and discrete \cite{BIC9} input distributions,
the secrecy performance for indoor VLC was discussed.
Analytical results suggest that the discrete input distribution significantly outperforms the truncated Gaussian and uniform distributions.
However, the discrete input distribution is still sub-optimal.
To further improve secrecy performance,
a better input distribution was obtained in \cite{BIB13}.
Focusing on a hybrid VLC/RF communication system with energy harvesting, the secrecy outage probability (SOP) was analyzed in \cite{BIB13_1}.
For VLC with spatially random terminals, the average secrecy capacity and the SOP were discussed in \cite{BIB13_2}.
To the best of our knowledge, the secrecy performance for the SM based VLC has not been well studied in open literature.

Motivated by the above work,
this paper analyzes the secrecy performance for an SM based VLC network,
which is consisted of multiple transmitters, a legitimate receiver, and an eavesdropper.
The main contributions are listed as follows.
\begin{itemize}
  \item The secrecy rate for SM based VLC with non-negativity and average optical intensity constraints is analyzed.
        By using the uniform selection (US) scheme and the existing results \cite{BIB13},
        a lower bound on secrecy rate is derived. According to the dual expression of secrecy rate, an upper bound of the secrecy rate is obtained. Both the lower and upper bounds are in closed-forms. Numerical results verify the tightness of these two newly derived bounds.
  \item The secrecy rate for SM based VLC with non-negativity, average optical intensity, and peak optical intensity constraints is analyzed.
        By adding the peak optical intensity constraint, the closed-form expressions of the secrecy rate bounds are further derived.
        The tightness of the lower and upper bounds are confirmed by numerical results.
  \item The asymptotic behaviors of the secrecy rate at high signal-to-noise ratio (SNR) are analyzed.
        At high SNR, the performance gap between the lower and upper bounds is small.
        Moreover, when the number of transmitter is one, the SM vanishes, the secrecy rate results coincide with the results in \cite{BIB13}.
  \item To improve the secrecy performance, a channel adaptive selection (CAS) scheme and a greedy selection (GS) scheme are proposed to select the active transmitter. Numerical results show that the proposed GS scheme performs better than the CAS and US schemes.
\end{itemize}

The rest of this paper is presented as follows.
Section \ref{section2} shows the system model.
In Sections \ref{section3} and \ref{section4},
the secrecy rate bounds and the asymptotic behaviors for the SM based VLC are analyzed over two scenarios.
Section \ref{section04} provides two transmitter selection schemes to improve secrecy performance.
Some typical numerical examples are given in Section \ref{section5}.
Finally, Section \ref{section6} concludes the paper and provides some future research directions.


\section{System Model}
\label{section2}
As illustrated in Fig. \ref{fig1}, an indoor VLC system with $M$ transmitters (i.e., Alice),
a legitimate receiver (i.e., Bob), and an eavesdropper (i.e., Eve) is considered.
For Alice, each transmitter employs an LED as the lighting source,
which is installed on the ceiling.
At Alice, the SM is employed, i.e., only one LED is active at each time instant and the others are silent.
The diagram of the SM in VLC is shown in Fig. \ref{fig1_1}.
At the receiver side, both Bob and Eve are located on the ground,
and each of them employs a photodiode (PD) to perform the optical-to-electrical conversion.
When an active LED transmits information to Bob, Eve can also receive the signal.

\begin{figure}
\centering
\includegraphics[width=8cm]{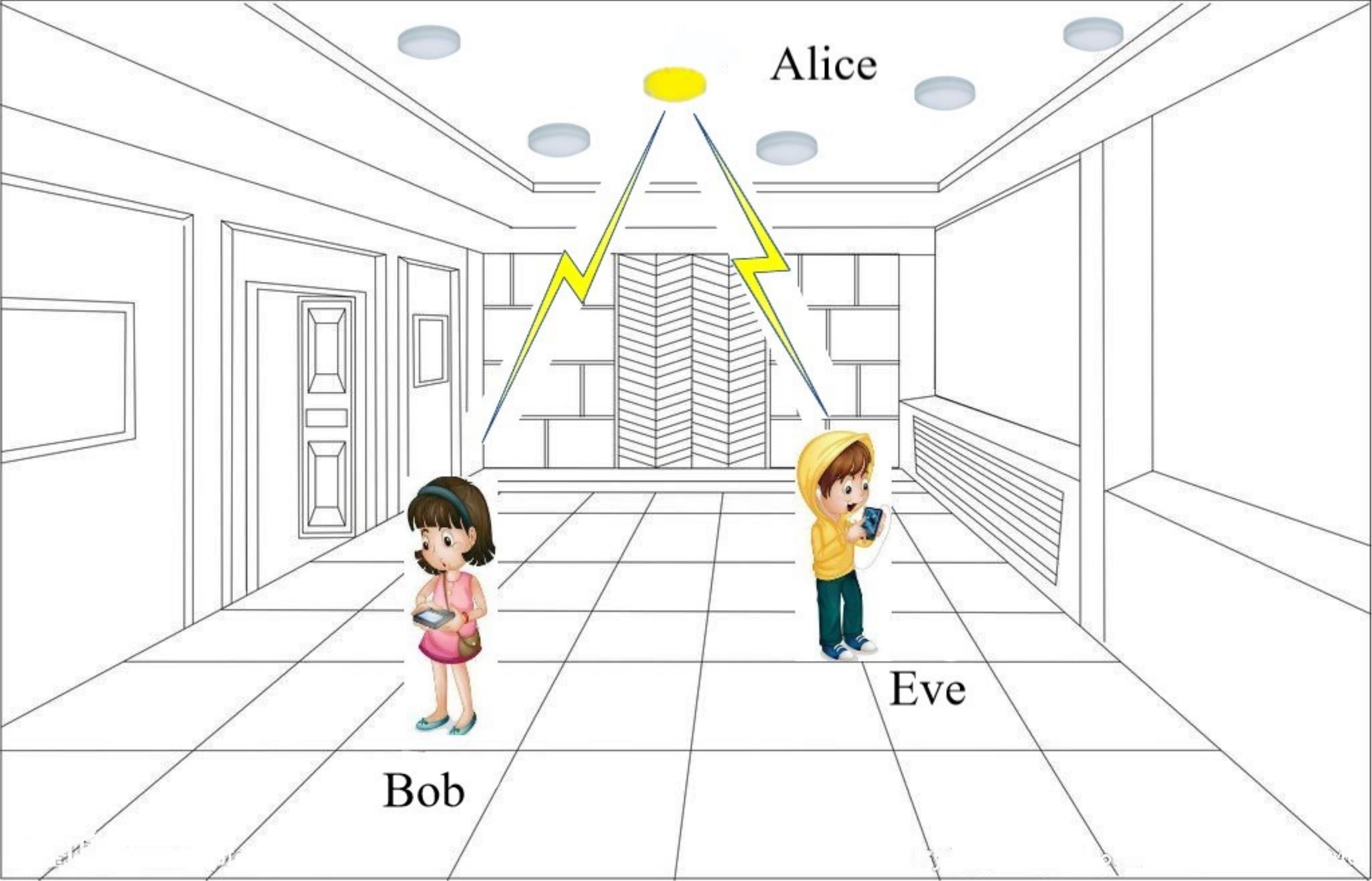}
\caption{An indoor VLC network with Alice, Bob and Eve.}
 \label{fig1}
\end{figure}

\begin{figure}
\centering
\includegraphics[width=8.5cm]{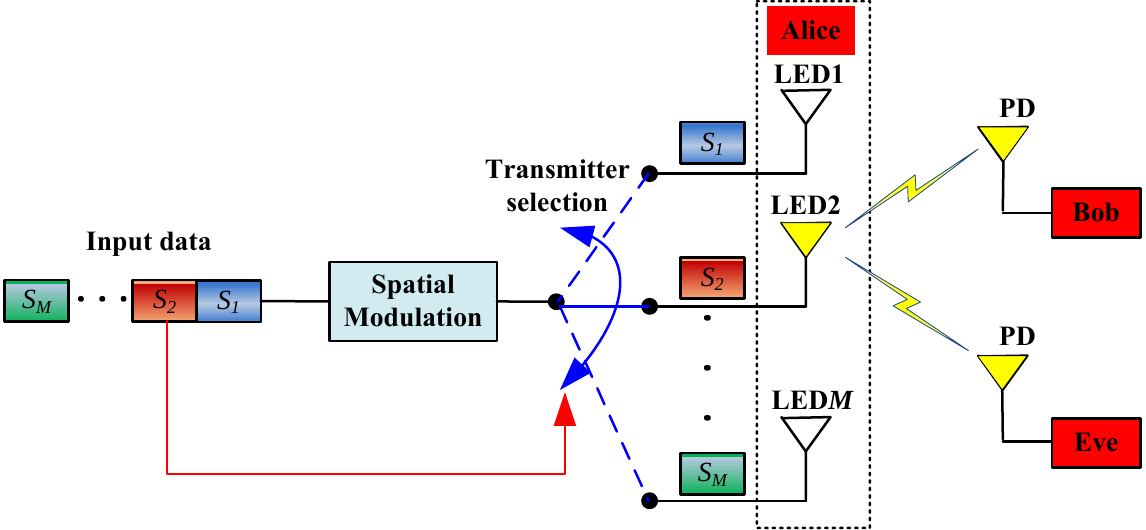}
\caption{The diagram of the SM in VLC.}
 \label{fig1_1}
\end{figure}

At the current time instant, we suppose that the $m$-th LED is activated.
Therefore, the received signals at Bob and Eve are given by
\begin{equation}
\left\{ \begin{array}{l}
{Y_{\rm{B}}} = {h_{{\rm{B,}}m}}X + {Z_{\rm{B}}}\\
{Y_{\rm{E}}} = {h_{{\rm{E,}}m}}X + {Z_{\rm{E}}}
\end{array} \right., m=1,2,\cdots,M,
 \label{eq1}
\end{equation}
where ${Z_{\rm{B}}} \sim N(0,\sigma _{\rm{B}}^2)$ and ${Z_{\rm{E}}} \sim N(0,\sigma _{\rm{E}}^2)$ are additive white Gaussian noises at Bob and Eve, $\sigma _{\rm{B}}^2$ and $\sigma _{\rm{E}}^2$ are the noise variances.
${h_{k,m}}$ is the direct current channel gain between the $m$-th LED and the $k$-th receiver ( $k = {\rm{B}}$ for Bob and $k = {\rm{E}}$  for Eve),
which is given by (6) in \cite{BIB14}.

By using SM, only one LED is activated at each time instant.
To choose the active LED, the uniform selection (US) scheme is utilized,
i.e., selecting each LED is equi-probable.
Therefore, the probability $p({h_k} = {h_{k,m}})$ can be expressed as
\begin{equation}
p({h_k} = {h_{k,m}}) = 1/M,\; k={\rm B\; or\; E}.
 \label{eq4}
\end{equation}

In (\ref{eq1}), the input signal $X$ is the transmitted optical intensity signal, which satisfies
the non-negativity constraint, i.e.,
\begin{equation}
X \ge 0.
 \label{eq3_1}
\end{equation}

Considering its physical characteristics, the LED is limited by its peak optical intensity $A$.
Consequently, the peak optical intensity constraint can be presented as
\begin{equation}
X \leq A.
 \label{eq3_2}
\end{equation}

To satisfy the illumination requirement in indoor scenario,
the dimmable average optical intensity constraint should be considered.
In VLC, the average optical intensity constraint depends on the dimming target,
which is given by \cite{BIB15}
\begin{equation}
{\mathbb{E}}(X) = \xi P,
 \label{eq3_3}
\end{equation}
where $\mathbb{E}( \cdot )$ is the expectation operator, $\xi  \in (0,1)$ denotes the dimming target,
and $P \in (0, A]$ represents the nominal optical intensity of each LED.

\section{Secrecy Rate for SM Based VLC with Constraints (\ref{eq3_1}) and (\ref{eq3_3})}
\label{section3}
Under constraints (\ref{eq3_1}) and (\ref{eq3_3}), the secrecy rate bounds for SM based VLC will be analyzed in this section,
The asymptotic behaviors of the secrecy rate at high SNR will be presented.

According to information theory \cite{BIB16},
when the main channel is inferior to the eavesdropper's channel (i.e., $h_{{\rm B},m}/\sigma_{\rm B} < h_{{\rm E},m}/\sigma_{\rm E} $),
the secrecy rate is zero;
otherwise, a positive secrecy rate ${R_s}$ for SM based VLC with constraints (\ref{eq3_1}) and (\ref{eq3_3}) is derived by solving the following problem
\begin{eqnarray}
&&{R_s} = \mathop {\max }\limits_{{f_X}\left( x \right)} \left[ {{\cal I}(X,{h_{\rm{B}}};{Y_{\rm{B}}}) - {\cal I}(X,{h_{\rm{E}}};{Y_{\rm{E}}})} \right] \nonumber \\
{\rm{s}}{\rm{.t}}{\rm{.}}&&\int_0^\infty  {{f_X}\left( x \right){\rm{d}}x}  = 1 \nonumber \\
&&{\mathbb{E}}\left( X \right) = \int_0^\infty  {x{f_X}\left( x \right){\rm{d}}x}  = \xi P,
 \label{eq5}
\end{eqnarray}
where ${f_X}\left( x \right)$ is the probability density function (PDF) of $X$,
${\cal I}( \cdot ; \cdot )$ denotes the mutual information.
Note that it is extremely challenging to solve optimization problem (\ref{eq5}).
Alternatively, tight secrecy rate bounds will be analyzed in the following.

\subsection{Lower Bound of Secrecy Rate}
\label{section3_1}
By analyzing optimization problem (\ref{eq5}),
a lower bound on secrecy rate for SM based VLC with constraints (\ref{eq3_1}) and (\ref{eq3_3}) is obtained in the following theorem.

\begin{theorem}
For the SM based VLC with constraints (\ref{eq3_1}) and (\ref{eq3_3}), the secrecy rate is lower-bounded by
\begin{equation}
{R_s} \!\ge\! \frac{1}{{2M}}\sum\limits_{m = 1}^M {\ln \left(\frac{\sigma_{\rm E}^2}{2\pi e \sigma_{\rm B}^2} {\frac{{{e^2}h_{{\rm{B}},m}^2{\xi ^2}{P^2} + 2\pi e\sigma _{\rm{B}}^2}}{{h_{{\rm{E}},m}^2{\xi ^2}{P^2} + \sigma _{\rm{E}}^2}}} \right)}.
 \label{eq6}
\end{equation}
\label{them1}
\end{theorem}

\begin{proof}
See Appendix \ref{appa}.
\end{proof}

\begin{corollary}
 When the number of LEDs is one (i.e., $M = 1$), the SM scheme vanishes, and the secrecy rate bound in (\ref{eq6})
coincides with (8) in \cite{BIB13}.
\label{coro1}
\end{corollary}

\subsection{Upper Bound of Secrecy Rate}
\label{section3_2}
In this subsection,
the dual expression of the secrecy rate \cite{BIB19,BIB17,BIB20} is adopted to further analyze the upper bound on the secrecy rate.

To facilitate the derivation, eq. (\ref{eq1}) can be re-formulated as
\begin{equation}
\left\{ \begin{array}{l}
Y'_{{\rm B},m} = X + Z'_{{\rm B},m}\\
Y'_{{\rm E},m} = X + Z'_{{\rm E},m}
\end{array} \right.,
 \label{eq1_1}
\end{equation}
where $Y'_{{\rm B},m}=Y_{\rm B}/h_{{\rm B},m}$, $Y'_{{\rm E},m}=Y_{\rm E}/h_{{\rm E},m}$,
$Z'_{{\rm B},m}=Z_{\rm B}/h_{{\rm B},m}$, and $Z'_{{\rm E},m}=Z_{\rm E}/h_{{\rm E},m}$.

\begin{lemma}
The conditional mutual information ${\cal I}\left( {X;{Y'_{{\rm B},m}}\left| {{Y'_{{\rm E},m}}} \right.} \right)$ is upper-bounded by
\begin{equation}
{\cal I}\left( {X;{Y'_{{\rm B},m}}\left| {{Y'_{{\rm E},m}}} \right.} \right) \le {{\mathbb{E}}_{X{Y'_{{\rm E},m}}}}\left\{ u \right\},
\label{eq21}
\end{equation}
where $u$ denotes a relative entropy, and it is defined as
\begin{eqnarray}
 u = D\left({{f_{{Y'_{{\rm B},m}}\left| {X{Y'_{{\rm E},m}}} \right.}}\left( {{y'_{{\rm B},m}}\left| {X,{Y'_{{\rm E},m}}} \right.} \right)\left\| {{g_{{Y'_{{\rm B},m}}\left| {{Y'_{{\rm E},m}}} \right.}}\left( {{y'_{{\rm B},m}}\left| {{Y'_{{\rm E},m}}}  \right.} \right)} \right.} \right),
\end{eqnarray}
where ${g_{{Y'_{{\rm B},m}}\left| {{Y'_{{\rm E},m}}} \right.}}\left( {{y'_{{\rm B},m}}\left| {{Y'_{{\rm E},m}}} \right.} \right)$ is an arbitrary conditional PDF of $Y'_{{\rm B},m}$ given $Y'_{{\rm E},m}$.
\label{le1}
\end{lemma}

\begin{proof}
See Appendix \ref{appc}.
\end{proof}

From \emph{Lemma \ref{le1}},
it can be observed that selecting an arbitrary ${g_{{Y'_{{\rm B},m}}\left| {{Y'_{{\rm E},m}}} \right.}}\left( {{y'_{{\rm B},m}}\left| {{Y'_{{\rm E},m}}} \right.} \right)$ in (\ref{eq21}) will result in an upper bound of ${\cal I}\left( {X;{Y'_{{\rm B},m}}\left| {{Y'_{{\rm E},m}}} \right.} \right)$.
Therefore, we have
\begin{equation}
{\cal I}\left( {X;{Y'_{{\rm B},m}}\left| {{Y'_{{\rm E},m}}} \right.} \right) \!=\! \mathop {\min }\limits_{{g_{{Y'_{{\rm B},m}}\!\left| {{Y'_{{\rm E},m}}} \right.}}\!\!\!\left( {{y'_{{\rm B},m}}\left| {{Y'_{{\rm E},m}}} \right.} \!\!\right)} {{\mathbb{E}}_{X{Y'_{{\rm E},m}}}}\!\!\!\left\{ u \right\}.
 \label{eq22}
\end{equation}

According to (\ref{eq1_1}) and (\ref{eq7}), $R_s$ can be re-expressed as
\begin{eqnarray}
{R_s} 
=\mathop {\max }\limits_{{f_X}\left( x \right)} \frac{1}{M} \sum\limits_{m = 1}^M {\cal I}(X;Y'_{{\rm B},m}|Y'_{{\rm E},m}).
\label{eq22_1}
\end{eqnarray}
Note that a unique input PDF ${f_{{X^*}}}\left( x \right)$ can be found to maximize $\sum_{m = 1}^M {\cal I}(X;Y'_{{\rm B},m}|Y'_{{\rm E},m})/M$ under constraints (\ref{eq3_1}) and (\ref{eq3_3}).
Therefore, $R_s$ in (\ref{eq22_1}) can be further written as \cite{BIC7}
\begin{equation}
{R_s} = \frac{1}{M} \sum\limits_{m = 1}^M \left\{\mathop {\min }\limits_{{g_{{Y'_{{\rm B},m}}\!\left| {{Y'_{{\rm E},m}}} \right.}}\!\!\!\left( {{y'_{{\rm B},m}}\left| {{Y'_{{\rm E},m}}} \right.} \!\!\right)}  {{\mathbb{E}}_{{X^*}{Y'_{{\rm E},m}}}}\left\{ u \right\}\right\},
 \label{eq23}
\end{equation}
where ${X^*}$ and ${f_{{X^*}}}\left( x \right)$ denote the optimal input and its PDF.

Consequently, we can get an upper bound of the secrecy rate as \cite{BIC7}
\begin{equation}
{R_s} \le \frac{1}{M} \sum\limits_{m = 1}^M  {{\mathbb{E}}_{{X^*}{Y'_{{\rm E},m}}}}\left\{ u \right\}.
 \label{eq24}
\end{equation}
By analyzing (\ref{eq24}), \emph{Theorem \ref{them2}} is obtained as follows.

\begin{theorem}
For the SM based VLC with constraints (\ref{eq3_1}) and (\ref{eq3_3}), the secrecy rate is upper-bounded by
\begin{eqnarray}
{R_s} \le  {\left\{ \begin{array}{l}
\frac{1}{M}\sum\limits_{m = 1}^M \ln \left[ {\frac{{4e\left( {\frac{{{\sigma _{\rm{B}}}}}{{\sqrt {2\pi } }} + \frac{{{h_{{\rm{B}},m}}\xi P}}{2}} \right)}}{{\sqrt {2\pi e\sigma _{\rm{B}}^2\left( {{\rm{1 + }}\frac{{\sigma _{\rm{B}}^2h_{{\rm{E}},m}^2}}{{\sigma _{\rm{E}}^2h_{{\rm{B}},m}^2}}} \right)} }}} \right],\;{\rm{if}}\;\sqrt {\frac{{\frac{{\sigma _{\rm{B}}^2}}{{h_{{\rm{B}},m}^2}} + \frac{{\sigma _{\rm{E}}^2}}{{h_{{\rm{E}},m}^2}}}}{{2\pi }}}  \ge \frac{{{\sigma _{\rm{B}}}}}{{\sqrt {2\pi } {h_{{\rm{B}},m}}}} + \frac{{\xi P}}{2}\\
\frac{1}{M}\sum\limits_{m = 1}^M \ln \left( {\frac{{2\sqrt e {h_{{\rm{B}},m}}{\sigma _{\rm{E}}}}}{{\pi {h_{{\rm{E}},m}}{\sigma _{\rm{B}}}}}} \right),\;{\rm{otherwise}}
\end{array} \right.}
 \label{eq25}
\end{eqnarray}
\label{them2}
\end{theorem}

\begin{proof}
See Appendix \ref{appb}.
\end{proof}

\begin{corollary}
When the number of LEDs $M$ is one, eq. (\ref{eq25})
is the same as (16) in \cite{BIB13}.
\label{coro2}
\end{corollary}

\subsection{Asymptotic Behavior Analysis}
\label{section3_3}
In a typical indoor VLC scenario, the received SNR is large (generally greater than 30 dB).
Therefore, we are more interested in the secrecy rate in the high SNR regime.
In this subsection, under constraints (\ref{eq3_1}) and (\ref{eq3_3}),
we analyze the asymptotic behaviors of the upper and lower bounds of the secrecy rate when $P$ tends to infinity.

By analyzing \emph{Theorem \ref{them1}}, when $P \to \infty $, we have
\begin{equation}
\mathop {\lim }\limits_{P \to \infty } {R_s} \ge \frac{1}{2} \ln \left(\frac{e}{2\pi} \right) + \frac{1}{2M} \sum\limits_{m = 1}^M   \ln \left( {\frac{{\sigma _{\rm{E}}^2h_{{\rm{B}},m}^2}}{{\sigma _{\rm{B}}^2}h_{{\rm{E}},m}^2}} \right).
 \label{eq42}
\end{equation}

By analyzing \emph{Theorem \ref{them2}}, when $P \to \infty $,
we have
\begin{equation}
\mathop {\lim }\limits_{P \to \infty } {R_s} \le \frac{1}{2} \ln \left(\frac{4e}{\pi^2} \right)+ \frac {1}{2M} \sum\limits_{m = 1}^M  \ln \left( {\frac{{{\sigma _{\rm{E}}^2} {h_{{\rm{B}},m}^2}}}{{{\sigma _{\rm{B}}^2} {h_{{\rm{E}},m}^2}}}} \right).
 \label{eq43}
\end{equation}

\begin{corollary}
For the SM based VLC under constraints (\ref{eq3_1}) and (\ref{eq3_3}),
the asymptotic behavior of the secrecy rate bounds in the high SNR regime is expressed as
\begin{eqnarray}
\frac{1}{2} \ln \left(\frac{e}{2\pi} \right) + \frac{1}{2M} \sum\limits_{m = 1}^M   \ln \left( {\frac{{\sigma _{\rm{E}}^2h_{{\rm{B}},m}^2}}{{\sigma _{\rm{B}}^2}h_{{\rm{E}},m}^2}} \right)  \le \mathop {\lim }\limits_{P \to \infty } {R_s} \le \frac{1}{2} \ln \left(\frac{4e}{\pi^2} \right)+ \frac {1}{2M} \sum\limits_{m = 1}^M  \ln \left( {\frac{{{\sigma _{\rm{E}}^2} {h_{{\rm{B}},m}^2}}}{{{\sigma _{\rm{B}}^2} {h_{{\rm{E}},m}^2}}}} \right).
\label{eq44}
\end{eqnarray}
\label{coro3}
\end{corollary}

\begin{remark}
In \emph{Corollary \ref{coro3}}, the asymptotic lower and upper bounds on secrecy rate do not coincide,
and their difference is $0.5\ln[4e/(\pi^2)]-0.5\ln[e/(2\pi)] \approx 0.4674$ nat/transmission.
In other words, the asymptotic performance gap is small.
\label{rem1}
\end{remark}

\section{Secrecy Rate for SM Based VLC with Constraints (\ref{eq3_1}), (\ref{eq3_2}) and (\ref{eq3_3})}
\label{section4}
By adding a peak optical intensity constraint (\ref{eq3_2}),
the secrecy rate bounds and the asymptotic behaviors at high SNR for the SM based VLC will be further analyzed.

Similarly, when $h_{{\rm B},m}/\sigma_{\rm B} < h_{{\rm E},m}/\sigma_{\rm E} $, the secrecy rate is zero.
When $h_{{\rm B},m}/\sigma_{\rm B} \geq h_{{\rm E},m}/\sigma_{\rm E} $, the secrecy rate for SM based VLC with constraints (\ref{eq3_1}), (\ref{eq3_2}) and (\ref{eq3_3}) derived by solving
\begin{eqnarray}
&&{R_s} = \mathop {\max }\limits_{{f_X}\left( x \right)} \left[ {{\cal I}(X,{h_{\rm{B}}};{Y_{\rm{B}}}) - {\cal I}(X,{h_{\rm{E}}};{Y_{\rm{E}}})} \right] \nonumber \\
{\rm{s}}{\rm{.t}}{\rm{.}}&&\int_0^A  {{f_X}\left( x \right){\rm{d}}x}  = 1 \nonumber \\
&&{\mathbb{E}}\left( X \right) = \int_0^A  {x{f_X}\left( x \right){\rm{d}}x}  = \xi P.
\label{p2}
\end{eqnarray}
Note that it is also challenging to obtain a closed-form solution for problem (\ref{p2}).
Similarly, tight secrecy rate bounds will be analyzed in the following subsections.

\subsection{Lower Bound of Secrecy Rate}
\label{section4_1}
At first, we define the average to peak optical intensity ratio as $\alpha= \xi P/A$.
By analyzing problem (\ref{p2}),
a lower bound on secrecy rate for SM based VLC with constraints (\ref{eq3_1}), (\ref{eq3_2}) and (\ref{eq3_3}) is obtained in the following theorem.

\begin{theorem}
For the SM based VLC with constraints (\ref{eq3_1}), (\ref{eq3_2}) and (\ref{eq3_3}), the secrecy rate is lower-bounded by
\begin{equation}
{R_s} \ge \left\{ \begin{array}{l}
\frac{1}{{2M}} \sum\limits_{m = 1}^M {\ln \left[ \frac{3\sigma_{\rm E}^2 \left(A^2 {h_{{\rm{B}},m}^2} + 2\pi e\sigma _{\rm{B}}^2\right)}{{2\pi e \sigma_{\rm B}^2 \left( h_{{\rm E},m}^2 \xi^2 P^2 + 3\sigma_{\rm E}^2 \right)}} \right]},\; {\rm if}\; \alpha=0.5\\
\frac{1}{{2M}}\sum\limits_{m = 1}^M \ln \left\{ \frac{\sigma_{\rm E}^2 \left[h_{{\rm B},m}^2 e^{-2c \xi P} \left(\frac{e^{cA}-1}{c} \right)^2 + 2 \pi e \sigma_{\rm B}^2 \right] }{2 \pi e \sigma_{\rm B}^2 \left[\frac{h_{{\rm E},m}^2  A(cA-2)}{c(1-e^{-cA})} + \frac{2 h_{{\rm E},m}^2 }{c^2}- h_{{\rm E},m}^2 \xi^2 P^2 + \sigma_{\rm E}^2 \right] }\right\},{\rm if}\; \alpha \neq 0.5\; {\rm and}\; \alpha \in (0,1]
\end{array} \right.
\label{equ1}
\end{equation}
where $c$ can be obtained by solving the following equation
\begin{equation}
\alpha  = \frac{1}{{1 - {e^{ - cA}}}} - \frac{1}{{cA}}.
\label{equ2}
\end{equation}
\label{them3}
\end{theorem}

\begin{proof}
See Appendix \ref{appd}.
\end{proof}

\begin{corollary}
When the number of LEDs $M$ is one, eq. (\ref{equ1})
is the same as (20) in \cite{BIB13}.
\label{coro4}
\end{corollary}

\subsection{Upper Bound of Secrecy Rate}
\label{section4_2}
In this subsection,
the dual expression of the secrecy rate \cite{BIB19,BIB17,BIB20} is also utilized to analyze the upper bound of the secrecy rate.
For this scenario, eq. (\ref{eq24}) can also be derived.
By analyzing (\ref{eq24}), \emph{Theorem \ref{them4}} is obtained.

\begin{theorem}
For the SM based VLC with constraints (\ref{eq3_1}), (\ref{eq3_2}) and (\ref{eq3_3}), the secrecy rate is upper-bounded by
\begin{eqnarray}
{R_s} \le \frac{{\rm{1}}}{{{\rm{2}}M}}\!\!\sum\limits_{m = 1}^M {\ln \!\!\left[\! {\frac{{\left( {\frac{{h_{{\rm{E}},m}^2}}{{h_{{\rm{B}},m}^2}}\sigma _{\rm{B}}^{\rm{2}}{\rm{ + }}\sigma _{\rm{E}}^{\rm{2}}} \right)\left( {h_{{\rm{B}},m}^2A\xi P{\rm{ + }}\sigma _{\rm{B}}^{\rm{2}}} \right)}}{{\sigma _{\rm{B}}^{\rm{2}}\!\!\left(\! {h_{{\rm{E}},m}^2\!A\xi P \!+\! 2\frac{{h_{{\rm{E}},m}^2}}{{h_{{\rm{B}},m}^2}}\sigma _{\rm{B}}^{\rm{2}} \!+\! \sigma _{\rm{E}}^{\rm{2}}}\! \right)\!\!\!\left(\! {1 \!+\! \frac{{h_{{\rm{E}},m}^2\sigma _{\rm{B}}^{\rm{2}}}}{{h_{{\rm{B}},m}^2\sigma _{\rm{E}}^{\rm{2}}}}}\! \right)}}}\!\! \right]}.
\label{equ4}
\end{eqnarray}
\label{them4}
\end{theorem}

\begin{proof}
See Appendix \ref{appe}.
\end{proof}

\begin{corollary}
 When the number of LEDs $M$ is one, eq. (\ref{equ4}) reduces to
(26) in \cite{BIB13}.
\label{coro5}
\end{corollary}

\subsection{Asymptotic Behavior Analysis}
\label{section4_3}
In subsections \ref{section4_1} and \ref{section4_2},
the lower and upper bounds on secrecy rate for the SM based VLC with constraints (\ref{eq3_1}), (\ref{eq3_2}) and (\ref{eq3_3}) are derived.
In this subsection, the asymptotic behavior of the secrecy rate when $A$ tends to infinity will be analyzed.

By analyzing \emph{Theorem \ref{them3}}, when $\alpha=0.5$, we have $A=2\xi P$.
Therefore, eq. (\ref{equ1}) can be further written as
\begin{equation}
{R_s} \ge\frac{1}{{2M}} \sum\limits_{m = 1}^M \!\!{\ln \left[ \frac{3\sigma_{\rm E}^2 \left(A^2 {h_{{\rm{B}},m}^2} + 2\pi e\sigma _{\rm{B}}^2\right)}{{2\pi e \sigma_{\rm B}^2 \left( h_{{\rm E},m}^2 \frac{A^2}{4}  + 3\sigma_{\rm E}^2 \right)}} \right]}.
\end{equation}
Then, we can get
\begin{equation}
\mathop {\lim }\limits_{A \to \infty } {R_s} \ge \frac{1}{M}\sum\limits_{m = 1}^M {\ln \left(\sqrt{\frac{6}{\pi e}} {\frac{{{h_{{\rm{B}},m}}{\sigma _{\rm{E}}}}}{{{h_{{\rm{E}},m}}{\sigma _{\rm{B}}}}}} \right)}.
 \label{equ6}
\end{equation}

By analyzing \emph{Theorem \ref{them4}}, when $A \to \infty$, we have
 \begin{equation}
\mathop {\lim }\limits_{A \to \infty } {R_s} \le \frac{1}{M}\sum\limits_{m = 1}^M {\ln \left( {\frac{{{h_{{\rm{B}},m}}{\sigma _{\rm{E}}}}}{{{h_{{\rm{E}},m}}{\sigma _{\rm{B}}}}}} \right)}.
 \label{equ7}
\end{equation}

\begin{corollary}
For the SM based VLC with constraints (\ref{eq3_1}), (\ref{eq3_2}) and (\ref{eq3_3}),
the asymptotic behavior of the secrecy rate when $\alpha=0.5$ is given by
\begin{eqnarray}
\frac 12 \ln \left( {\frac{{\rm{6}}}{{\pi e}}}\right) + \frac{1}{M}\sum\limits_{m = 1}^M {\ln \left( { \frac{{{h_{{\rm{B}},m}}{\sigma _{\rm{E}}}}}{{{h_{{\rm{E}},m}}{\sigma _{\rm{B}}}}}} \right)}  \le \mathop {\lim }\limits_{A \to \infty } {R_s} \le \frac{1}{M}\sum\limits_{m = 1}^M {\ln \left( {\frac{{{h_{{\rm{B}},m}}{\sigma _{\rm{E}}}}}{{{h_{{\rm{E}},m}}{\sigma _{\rm{B}}}}}} \right)}.
\label{equ8}
\end{eqnarray}
\label{coro6}
\end{corollary}

\begin{remark}
In \emph{Corollary \ref{coro6}}, when $\alpha=0.5$,
the asymptotic performance gap equals $0.5\ln[{\pi e/6 }] \approx 0.1765$ nat/transmission.
Although a performance gap exists between the asymptotic lower and upper bounds, the difference is small.
\label{rem2}
\end{remark}

When $\alpha \neq 0.5$ and $\alpha \in (0,1]$, we have $\xi P= \alpha A$.
Moreover, for any $\alpha$ in (\ref{equ2}), $cA$ is a constant.
Let $b=cA$, eq. (\ref{equ1}) can be further written as
\begin{eqnarray}
{R_s} \ge\frac{1}{{2M}}\sum\limits_{m = 1}^M \ln \left\{ {\frac{{\sigma _{\rm{E}}^{\rm{2}}h_{{\rm{B}},m}^2{e^{ - 2\alpha b}}\frac{{{{({e^b} - 1)}^2}}}{{{b^2}}}{A^2} + 2\pi e\sigma _{\rm{E}}^{\rm{2}}\sigma _{\rm{B}}^{\rm{2}}}}{{{\rm{2}}\pi e\sigma _{\rm{B}}^{\rm{2}}\left[ {\frac{{h_{{\rm{E}},m}^2(b - 2)}}{{b(1 - {e^{ - b}})}} + \frac{{2h_{{\rm{E}},m}^2}}{{{b^2}}} - h_{{\rm{E}},m}^2{\alpha ^2}} \right]{A^2} + {\rm{2}}\pi e\sigma _{\rm{B}}^{\rm{2}}\sigma _{\rm{E}}^2}}} \right\},
\label{equ9}
\end{eqnarray}
where $b$ can be derived by
\begin{equation}
\alpha  = \frac{1}{{1 - {e^{ - b}}}} - \frac{1}{b}.
\end{equation}

When $A \to \infty$, eq. (\ref{equ9}) can be written as
\begin{eqnarray}
\mathop {\lim }\limits_{A \to \infty } {R_s} \ge \ln \left\{ {\frac{{{e^{ - \alpha b}}({e^b} - 1)}}{{b\sqrt {{\rm{2}}\pi e\left[ {\frac{{b - 2}}{{b(1 - {e^{ - b}})}} + \frac{2}{{{b^2}}} - {\alpha ^2}} \right]} }}} \right\} + \frac{1}{M}\sum\limits_{m = 1}^M {\ln \left( {\frac{{{h_{{\rm{B}},m}}{\sigma _{\rm{E}}}}}{{{h_{{\rm{E}},m}}{\sigma _{\rm{B}}}}}} \right)}.
\label{equ10}
\end{eqnarray}
According to (\ref{equ7}) and (\ref{equ10}), the following corollary is obtained.

\begin{corollary}
For the SM based VLC with constraints (\ref{eq3_1}), (\ref{eq3_2}) and (\ref{eq3_3}),
the asymptotic behavior of the secrecy rate when $\alpha \neq 0.5$ and $\alpha \in (0,1]$ is given by
\begin{eqnarray}
\ln\!\! \left\{\!\! {\frac{{{e^{ - \alpha b}}({e^b} - 1)}}{{b\sqrt {{\rm{2}}\pi e\!\!\left[ {\frac{{b - 2}}{{b(1 - {e^{ - b}})}} \!+\! \frac{2}{{{b^2}}} \!-\! {\alpha ^2}} \right]} }}}\!\! \right\} \!+\! \frac{1}{M}\!\!\sum\limits_{m = 1}^M \!\!{\ln\!\! \left(\!\! {\frac{{{h_{{\rm{B}},m}}{\sigma _{\rm{E}}}}}{{{h_{{\rm{E}},m}}{\sigma _{\rm{B}}}}}}\!\! \right)} \le \mathop {\lim }\limits_{A \to \infty } {R_s} \le \frac{1}{M}\sum\limits_{m = 1}^M {\ln \left( {\frac{{{h_{{\rm{B}},m}}{\sigma _{\rm{E}}}}}{{{h_{{\rm{E}},m}}{\sigma _{\rm{B}}}}}} \right)}.
\label{equ11}
\end{eqnarray}
\label{coro7}
\end{corollary}

\begin{remark}
In \emph{Corollary \ref{coro7}}, when $\alpha \neq 0.5$ and $\alpha \in (0,1]$,
the asymptotic performance gap at high SNR is $\ln \left\{ {b{e^{\alpha b}}\sqrt {{\rm{2}}\pi e\left[ {\frac{{b - 2}}{{b(1 - {e^{ - b}})}} + \frac{2}{{{b^2}}} - {\alpha ^2}} \right]} /({e^b} - 1)} \right\}$ nat/transmission.
Numerical results in Section \ref{section5} will show that such a performance gap is small.
\label{rem3}
\end{remark}

\section{Secrecy Performance Improvement Schemes}
\label{section04}
In Section \ref{section2}, the US scheme is utilized to select an active transmitter,
i.e., the probability of selecting each LED is assumed to be the same.
However, such a selection scheme does not perform well in some cases.
To improve the secrecy rate, two novel transmitter selection schemes are provided in this section.

\subsection{Channel Adaptive Selection Scheme}
As it is known, the probability of selecting each LED depends on the CSI of both Bob and Eve.
The larger the difference between ${h_{{\rm{B,}}m}}/{\sigma _{\rm{B}}}$ and ${h_{{\rm{E,}}m}}/{\sigma _{\rm{E}}}$ is,
the larger the secrecy rate becomes.
To enhance the secrecy rate,
the LED with large ${h_{{\rm{B},}m}}/{\sigma _{\rm{B}}}-{h_{{\rm{E,}}m}}/{\sigma _{\rm{E}}}$ should be selected with large probability.
Therefore, the probability of selecting the $m$-th LED in (\ref{eq4}) is modified as
\begin{eqnarray}
p({h_k} = {h_{k,m}}) = \frac{{\frac{{{h_{{\rm{B}},m}}}}{{{\sigma _{\rm{B}}}}} - \frac{{{h_{{\rm{E}},m}}}}{{{\sigma _{\rm{E}}}}}}}{{\sum\limits_{j = 1}^M {\left( {\frac{{{h_{{\rm{B}},j}}}}{{{\sigma _{\rm{B}}}}} - \frac{{{h_{{\rm{E}},j}}}}{{{\sigma _{\rm{E}}}}}} \right)} }}.
\label{eqq1}
\end{eqnarray}

In practice, the process of selecting each LED in the CAS scheme is presented in Algorithm \ref{alg1}.
Based on the CAS scheme,
\emph{Theorems \ref{them1}} and \emph{\ref{them2}} can be updated as \emph{Theorem \ref{them5}}.

\begin{algorithm}[!h]
\caption{The CAS scheme}
\begin{algorithmic}[1]
\State Given the noise variances ${\sigma _{\rm{B}}}$, ${\sigma _{\rm{E}}}$, and the number of LEDs $M$.
\State Obtain the positions of Alice, Bob and Eve.
\State Compute the probability of selecting each LED by using (\ref{eqq1}).
\State Compute the cumulative probabilities ${q_i} = \sum\nolimits_{m = 1}^i {p({h_k} = {h_{k,m}})}, i=1,\cdots,M$.
\State Generate a random number $r$ in the range of $[0,1]$.
\If{$r < {q_1}$}
　　　\State The first LED is selected;
\Else \;\textbf{if} {${q_{k - 1}} < r \le {q_k}$} \textbf{then}
　　　\State The $k$-th LED is selected.
\EndIf
\State \textbf{endif}
\State Repeat Steps 2-10 to select another LED for the next time instant.
\end{algorithmic}
\label{alg1}
\end{algorithm}

\begin{theorem}
For the SM based VLC with constraints (\ref{eq3_1}) and (\ref{eq3_3}),
by using the CAS scheme in (\ref{eqq1}),
the lower and upper bounds of the secrecy rate are given by
\begin{eqnarray}
{R_s} \ge \frac{1}{2}\sum\limits_{m = 1}^M \left[ \frac{{\frac{{{h_{{\rm{B}},m}}}}{{{\sigma _{\rm{B}}}}} - \frac{{{h_{{\rm{E}},m}}}}{{{\sigma _{\rm{E}}}}}}}{{\sum\limits_{j = 1}^M {\left( {\frac{{{h_{{\rm{B}},j}}}}{{{\sigma _{\rm{B}}}}} - \frac{{{h_{{\rm{E}},j}}}}{{{\sigma _{\rm{E}}}}}} \right)} }}    \ln \left( {\frac{{\sigma _{\rm{E}}^{\rm{2}}}}{{{\rm{2}}\pi e\sigma _{\rm{B}}^2}}\frac{{{e^2}h_{{\rm{B}},m}^2{\xi ^2}{P^2} + 2\pi e\sigma _{\rm{B}}^2}}{{h_{{\rm{E}},m}^2{\xi ^2}{P^2} + \sigma _{\rm{E}}^2}}} \right) \right].
\end{eqnarray}
and
\begin{eqnarray}
{R_s} \le \left\{ \begin{array}{l}
\sum\limits_{m = 1}^M \!\!{\frac{{\frac{{{h_{{\rm{B}},m}}}}{{{\sigma _{\rm{B}}}}} - \frac{{{h_{{\rm{E}},m}}}}{{{\sigma _{\rm{E}}}}}}}{{\sum\limits_{j = 1}^M {\left( {\frac{{{h_{{\rm{B}},j}}}}{{{\sigma _{\rm{B}}}}} - \frac{{{h_{{\rm{E}},j}}}}{{{\sigma _{\rm{E}}}}}} \right)} }}\!\!\ln \!\!\left[ {\frac{{4e\left( {\frac{{{\sigma _{\rm{B}}}}}{{\sqrt {2\pi } }} + \frac{{{h_{{\rm{B}},m}}\xi P}}{2}} \right)}}{{\sqrt {2\pi e\sigma _{\rm{B}}^{\rm{2}}\left( {{\rm{1 + }}\frac{{\sigma _{\rm{B}}^{\rm{2}}h_{{\rm{E}},m}^2}}{{\sigma _{\rm{E}}^{\rm{2}}h_{{\rm{B}},m}^2}}} \right)} }}} \right],} \;{\rm{if}}\;\sqrt {\frac{{\frac{{\sigma _{\rm{B}}^2}}{{h_{{\rm{B}},m}^2}} + \frac{{\sigma _{\rm{E}}^2}}{{h_{{\rm{E}},m}^2}}}}{{2\pi }}}  \ge \frac{{{\sigma _{\rm{B}}}}}{{\sqrt {2\pi } {h_{{\rm{B}},m}}}} + \frac{{\xi P}}{2}\\
\!\!\!\!\sum\limits_{m = 1}^M\!\! {\frac{{\frac{{{h_{{\rm{B}},m}}}}{{{\sigma _{\rm{B}}}}} - \frac{{{h_{{\rm{E}},m}}}}{{{\sigma _{\rm{E}}}}}}}{{\sum\limits_{j = 1}^M {\left( {\frac{{{h_{{\rm{B}},j}}}}{{{\sigma _{\rm{B}}}}} - \frac{{{h_{{\rm{E}},j}}}}{{{\sigma _{\rm{E}}}}}} \right)} }}\ln \left( {\frac{{2\sqrt e {h_{{\rm{B}},m}}{\sigma _{\rm{E}}}}}{{\pi {h_{{\rm{E}},m}}{\sigma _{\rm{B}}}}}} \right)},\;{\rm{otherwise}}
\end{array} \right.
\end{eqnarray}
\label{them5}
\end{theorem}


By considering the CAS scheme in (\ref{eqq1}),
\emph{Theorems \ref{them3}} and \emph{\ref{them4}} can be modified as \emph{Theorem \ref{them6}}.
\begin{theorem}
For the SM based VLC with constraints (\ref{eq3_1}), (\ref{eq3_2}) and (\ref{eq3_3}),
by using the CAS scheme in (\ref{eqq1}),
the lower and upper bounds of the secrecy rate are given by
\begin{eqnarray}
{R_s} \!\!\ge\!\!\! \left\{ \begin{array}{l}\!\!\!\!
\frac{1}{2}\!\!\sum\limits_{m = 1}^M \!\!\!\left\{ \frac{{\frac{{{h_{{\rm{B}},m}}}}{{{\sigma _{\rm{B}}}}} - \frac{{{h_{{\rm{E}},m}}}}{{{\sigma _{\rm{E}}}}}}}{{\sum\limits_{j = 1}^M {\left( {\frac{{{h_{{\rm{B}},j}}}}{{{\sigma _{\rm{B}}}}} - \frac{{{h_{{\rm{E}},j}}}}{{{\sigma _{\rm{E}}}}}} \right)} }}  \ln \left[ {\frac{{2\sigma _{\rm{E}}^2\left( {{A^2}h_{{\rm{B}},m}^2 + 2\pi e\sigma _{\rm{B}}^2} \right)}}{{2\pi e\sigma _{\rm{B}}^2\left( {h_{{\rm{E}},m}^2{\xi ^{\rm{2}}}{P^2} + 3\sigma _{\rm{E}}^2} \right)}}} \right] \right\},\;{\rm{if}}\;\alpha  = 0.5\\
\!\!\!\!\frac{1}{2}\!\!\sum\limits_{m = 1}^M\!\!\! \left\{\!\! \frac{{\frac{{{h_{{\rm{B}},m}}}}{{{\sigma _{\rm{B}}}}} - \frac{{{h_{{\rm{E}},m}}}}{{{\sigma _{\rm{E}}}}}}}{{\sum\limits_{j = 1}^M \!\!\!{\left(\! {\frac{{{h_{{\rm{B}},j}}}}{{{\sigma _{\rm{B}}}}} - \frac{{{h_{{\rm{E}},j}}}}{{{\sigma _{\rm{E}}}}}}\! \right)} }} \! \ln\!\! \left[\!\! {\frac{{\sigma _{\rm{E}}^2\left[ {h_{{\rm{B}},m}^2{e^{ - 2c\xi P}}{{\left( {\frac{{{e^{cA}} - 1}}{c}} \right)}^2} + 2\pi e\sigma _{\rm{B}}^2} \right]}}{{2\pi e\sigma _{\rm{B}}^2\left( {\frac{{h_{{\rm{E}},m}^2A(cA - 2)}}{{c(1 - {e^{ - cA}})}} + \frac{{2h_{{\rm{E}},m}^2}}{{{c^2}}} - h_{{\rm{E}},m}^2{\xi ^2}{P^2} + \sigma _{\rm{E}}^{\rm{2}}} \right)}}}\!\! \right] \!\!\right\}\!\!,{\rm{if}}\;\alpha  \ne 0.5\;{\rm{and}}\;\alpha  \in (0,1]
\end{array} \right.
\end{eqnarray}
and
\begin{eqnarray}
{R_s} \le \frac{1}{2}\sum\limits_{m = 1}^M \left\{ \frac{{\frac{{{h_{{\rm{B}},m}}}}{{{\sigma _{\rm{B}}}}} - \frac{{{h_{{\rm{E}},m}}}}{{{\sigma _{\rm{E}}}}}}}{{\sum\limits_{j = 1}^M {\left( {\frac{{{h_{{\rm{B}},j}}}}{{{\sigma _{\rm{B}}}}} - \frac{{{h_{{\rm{E}},j}}}}{{{\sigma _{\rm{E}}}}}} \right)} }}  \ln \left[ {\frac{{\left( {\frac{{h_{{\rm{E}},m}^2}}{{h_{{\rm{B}},m}^2}}\sigma _{\rm{B}}^2 + \sigma _{\rm{E}}^2} \right)\left( {h_{{\rm{B}},m}^2A\xi P + \sigma _{\rm{B}}^{\rm{2}}} \right)}}{{\sigma _{\rm{B}}^2\left( {h_{{\rm{E}},m}^2A\xi P + 2\frac{{h_{{\rm{E}},m}^2}}{{h_{{\rm{B}},m}^2}}\sigma _{\rm{B}}^2 + \sigma _{\rm{E}}^2} \right)\left( {1 + \frac{{h_{{\rm{E}},m}^2\sigma_{\rm{B}}^2}}{{h_{{\rm{B}},m}^2\sigma_{\rm{E}}^2}}} \right)}}} \right] \right\}.
\end{eqnarray}
\label{them6}
\end{theorem}

\subsection{Greedy Selection Scheme}
In this subsection, the GS scheme is introduced.
In this scheme, the LED with the maximum value of $h_{{\rm B}, m}/\sigma_{\rm B} - h_{{\rm E}, m}/\sigma_{\rm E}$ is selected at each time instant.
Therefore, the probability of selecting the $m$-th LED is re-expressed as
\begin{eqnarray}
p({h_k} \!=\! {h_{k,m}}) \!=\!\! \left\{ \begin{array}{l}\!\!\!\!
1,\;{\rm if}\;m \!=\! \arg \mathop {\max }\limits_{k = 1,\cdots,M} \!\!\left\{\! {\frac{{{h_{{\rm{B}},k}}}}{{{\sigma _{\rm{B}}}}} \!-\! \frac{{{h_{{\rm E},k}}}}{{{\sigma _{\rm{E}}}}}}\! \right\}\\
\!\!\!\!0,\;{\rm otherwise}
\end{array} \right.
\label{eqw1}
\end{eqnarray}

For this scheme, the process of selecting each LED in the GS scheme is provided in Algorithm \ref{alg2}.
By using the GS scheme,
\emph{Theorems \ref{them1}} and \emph{\ref{them2}} can be updated as \emph{Theorem \ref{them7}}.

\begin{algorithm}[!h]
\caption{The GS scheme}
\begin{algorithmic}[1]
\State Given the noise variances ${\sigma _{\rm{B}}}$, ${\sigma _{\rm{E}}}$, and the number of LEDs $M$.
\State Obtain the positions of Alice, Bob and Eve.
\State Compute ${{h_{{\rm{B}},k}}{\rm{/}}{\sigma _{\rm{B}}} - {h_{{\rm E},k}}{\rm{/}}{\sigma _{\rm{E}}}}$ for $k=1,2,\cdots,M$.
\If{$m = \arg \mathop {\max }\limits_{k = 1, \cdots ,M} \left\{ {{h_{{\rm{B}},k}}{\rm{/}}{\sigma _{\rm{B}}} - {h_{{\rm E},k}}{\rm{/}}{\sigma _{\rm{E}}}} \right\}$}
　　　\State The $m$-th LED is selected.
\EndIf
\State \textbf{endif}
\State Repeat Steps 2-6 to select another LED for the next time instant.
\end{algorithmic}
\label{alg2}
\end{algorithm}

\begin{theorem}
For the SM based VLC with constraints (\ref{eq3_1}) and (\ref{eq3_3}),
by using the GS scheme in (\ref{eqw1}),
the lower and upper bounds of the secrecy rate are given by
\begin{eqnarray}
{R_s} \ge \mathop {\max }\limits_m \left\{ {\frac{1}{2}\ln \left( {\frac{{\sigma _{\rm{E}}^2}}{{2\pi e\sigma _{\rm{B}}^2}}\frac{{{e^2}h_{{\rm{B}},m}^2{\xi ^2}{P^2} + 2\pi e\sigma _{\rm{B}}^2}}{{h_{{\rm{E}},m}^2{\xi ^2}{P^2}{\rm{ + }}\sigma _{\rm{E}}^2}}} \right)} \right\}.
\end{eqnarray}
and
\begin{eqnarray}
{R_s} \!\le\! \left\{ \begin{array}{l}\!\!\!\!
\max\limits_{m} \left\{\ln \!\!\left[ {\frac{{4e\left( {\frac{{{\sigma _{\rm{B}}}}}{{\sqrt {2\pi } }} + \frac{{{h_{{\rm{B}},m}}\xi P}}{2}} \right)}}{{\sqrt {2\pi e\sigma _{\rm{B}}^{\rm{2}}\left( {{\rm{1 + }}\frac{{\sigma _{\rm{B}}^{\rm{2}}h_{{\rm{E}},m}^2}}{{\sigma _{\rm{E}}^{\rm{2}}h_{{\rm{B}},m}^2}}} \right)} }}} \right]\right\}, \;{\rm{if}}\;\sqrt {\frac{{\frac{{\sigma _{\rm{B}}^2}}{{h_{{\rm{B}},m}^2}} + \frac{{\sigma _{\rm{E}}^2}}{{h_{{\rm{E}},m}^2}}}}{{2\pi }}}  \ge \frac{{{\sigma _{\rm{B}}}}}{{\sqrt {2\pi } {h_{{\rm{B}},m}}}} + \frac{{\xi P}}{2}\\
\max\limits_{m} \left\{\ln \left( {\frac{{2\sqrt e {h_{{\rm{B}},m}}{\sigma _{\rm{E}}}}}{{\pi {h_{{\rm{E}},m}}{\sigma _{\rm{B}}}}}} \right)\right\},\;{\rm{otherwise}}
\end{array} \right.
\end{eqnarray}
\label{them7}
\end{theorem}

By using the GS scheme,
\emph{Theorems \ref{them3}} and \emph{\ref{them4}} can be updated as \emph{Theorem \ref{them8}}.
\begin{theorem}
For the SM based VLC with constraints (\ref{eq3_1}), (\ref{eq3_2}) and (\ref{eq3_3}),
by using the GS scheme in (\ref{eqw1}),
the lower and upper bounds of the secrecy rate are given by
\begin{eqnarray}
{R_s} \!\ge\!\! \left\{ \begin{array}{l}\!\!\!\!
\max\limits_{m} \left\{ \frac{1}{2} {  \ln \left[ {\frac{{2\sigma _{\rm{E}}^2\left( {{A^2}h_{{\rm{B}},m}^2 + 2\pi e\sigma _{\rm{B}}^2} \right)}}{{2\pi e\sigma _{\rm{B}}^2\left( {h_{{\rm{E}},m}^2{\xi ^{\rm{2}}}{P^2} + 3\sigma _{\rm{E}}^2} \right)}}} \right]}\right\},\;{\rm{if}}\;\alpha  = 0.5\\
\!\!\!\max\limits_{m}\!\left\{\!\frac{1}{2}  {\ln\! \left[\! {\frac{{\sigma _{\rm{E}}^2\left[ {h_{{\rm{B}},m}^2{e^{ - 2c\xi P}}{{\left( {\frac{{{e^{cA}} - 1}}{c}} \right)}^2} + 2\pi e\sigma _{\rm{B}}^2} \right]}}{{2\pi e\sigma _{\rm{B}}^2\!\!\left(\! {\frac{{h_{{\rm{E}},m}^2\!A(cA \!-\! 2)}}{{c(1 - {e^{ - cA}})}} + \frac{{2h_{{\rm{E}},m}^2}}{{{c^2}}} - h_{{\rm{E}},m}^2{\xi ^2}{P^2} + \sigma _{\rm{E}}^{\rm{2}}}\! \right)}}}\! \right]} \!\right\},{\rm{if}}\;\alpha  \ne 0.5\; {\rm{and}}\;\alpha  \in (0,1]
\end{array} \right.
\end{eqnarray}
and
\begin{eqnarray}
{R_s} \le \max\limits_{m} \left\{\frac{1}{2}  \ln \left[ {\frac{{\left( {\frac{{h_{{\rm{E}},m}^2}}{{h_{{\rm{B}},m}^2}}\sigma _{\rm{B}}^2 + \sigma _{\rm{E}}^2} \right)\left( {h_{{\rm{B}},m}^2A\xi P + \sigma _{\rm{B}}^{\rm{2}}} \right)}}{{\sigma _{\rm{B}}^2\left( {h_{{\rm{E}},m}^2A\xi P + 2\frac{{h_{{\rm{E}},m}^2}}{{h_{{\rm{B}},m}^2}}\sigma _{\rm{B}}^2 + \sigma _{\rm{E}}^2} \!\right)\left( {1+ \frac{{h_{{\rm{E}},m}^2\sigma_{\rm{B}}^2}}{{h_{{\rm{B}},m}^2\sigma_{\rm{E}}^2}}} \right)}}} \right]\right\}.
\end{eqnarray}
\label{them8}
\end{theorem}

\section{Numerical Results}
\label{section5}
In this section, some classic numerical evaluations of the secrecy rate for the indoor SM based VLC system are shown.
Here, we consider a typical three-node indoor VLC system with room size $5 \;{\rm m} \times 4 \;{\rm m} \times 3 \;{\rm m}$.
In the system, $M = 8$ LEDs (i.e., Alice) are installed on the ceiling with a height of 3 m,
while the legitimate receiver (Bob) and the eavesdropper (Eve) are located at the height of 0.8 m.
The coordinates of Alice and Bob are shown in Table \ref{tab1}.
In addition, the noise variances of both Bob and Eve are set to be $\sigma _{\rm{B}}^2 = \sigma _{\rm{E}}^2 =  - 104 \;\rm{dBm}$.

\begin{table}[!h]
\caption{Positions of Alice and Bob.}
\begin{center}
\begin{tabularx}{14cm}{|p{1.2cm}|X|p{1.8cm}|}
\hline\hline
\centering Scenarios  &\centering Alice  & \centering Bob
\tabularnewline\hline
\centering Positions &\centering (1, 2, 3), (1, 3, 3), (2, 2, 3), (2, 3, 3), (3, 2, 3), (3, 3, 3), (4, 2, 3), (4, 3, 3) & \centering (2.5, 1.5, 0.8)
\tabularnewline\hline\hline
\end{tabularx}
\end{center}
\label{tab1}
\end{table}

\begin{figure}
\centering
\includegraphics[width=8.5cm]{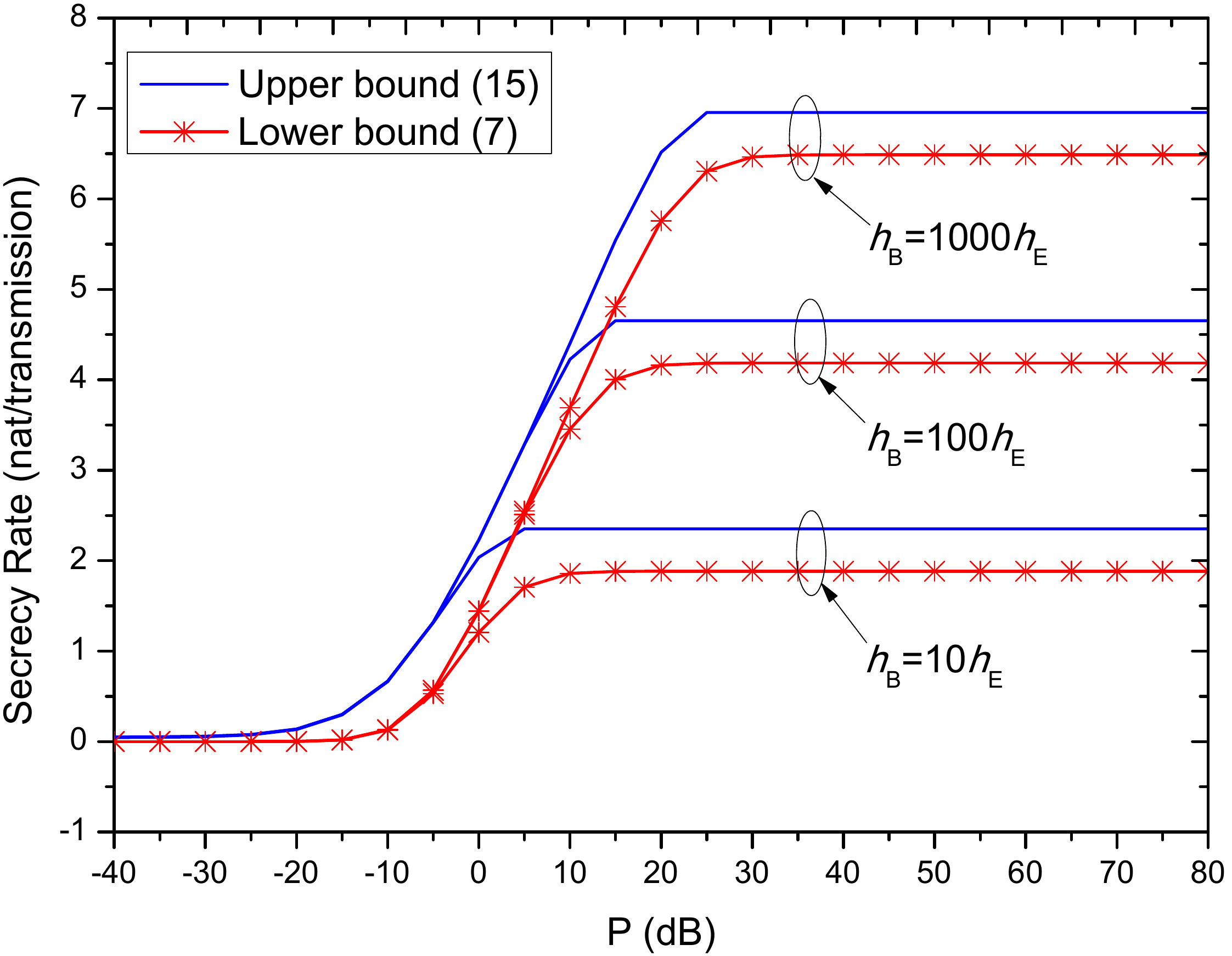}
\caption{Secrecy rate bounds versus $P$ with different ${h_{\rm{B}}}/{h_{\rm{E}}}$ when $\xi  = 0.5$.}
 \label{fig2}
\end{figure}

\subsection{Results of SM Based VLC with Constraints (\ref{eq3_1}) and (\ref{eq3_3})}

Fig. \ref{fig2} depicts the secrecy rate bounds versus $P$ with different ${h_{\rm{B}}}/{h_{\rm{E}}}$ when $\xi  = 0.5$.
It can be observed that, when ${h_{\rm{B}}}/{h_{\rm{E}}}$ takes a different value, the bounds of secrecy rate varies.
Specifically, as the increase of ${h_{\rm{B}}}/{h_{\rm{E}}}$, the secrecy rate bounds also increase.
Moreover, as the increase of $P$, all the secrecy rate bounds increase first and then tend to stable values.
To show the asymptotic behavior at high SNR, Table \ref{tab2} shows the secrecy performance gaps between lower bound (\ref{eq6}) and upper bound (\ref{eq25}).
The results show that the performance gap is about 0.4674 nat/transmission,
which coincides with the conclusion in \emph{Remark \ref{rem1}}.
In other words, the difference between the lower and upper bounds is small,
which demonstrates the correctness of the derived secrecy rate bounds.

\begin{table}
\caption{Performance gaps between (\ref{eq6}) and (\ref{eq25}) when $\xi  = 0.5$.}
\centering
\setlength{\tabcolsep}{6pt}
	\begin{tabular}{|c|c|c|c|}
		\hline\hline
        \multicolumn{1}{|c|}{ \multirow{2}*{$P$ ({\rm dB})} }& \multicolumn{3}{c|}{Performance gaps (nat/transmission)} \\
        \cline{2-4}
		{} & ${h_{\rm{B}}} = 10{h_{\rm{E}}}$ & ${h_{\rm{B}}} = 100{h_{\rm{E}}}$ & ${h_{\rm{B}}} = 1000{h_{\rm{E}}}$ \\
        \hline
        30 & 0.46736 & 0.46762 & 0.49213 \\
        \hline
        40 & 0.46735 & 0.46736 & 0.46761 \\
        \hline
        50 & 0.46735 & 0.46736 & 0.46736 \\
        \hline
        60 & 0.46735 & 0.46736 & 0.46735 \\
        \hline
        70 & 0.46735 & 0.46736 & 0.46735 \\
        \hline
        80 & 0.46735 & 0.46736 & 0.46735 \\
        \hline\hline
	\end{tabular}
\label{tab2}	
\end{table}

Fig. \ref{fig3} plots the relationship between the secrecy rate bounds and $\xi $ with different $P$ when ${h_{\rm{B}}}/{h_{\rm{E}}} = 1000$.
For small $\xi $, a rapid increase in the secrecy rate bounds can be observed with the increase of $\xi$.
However, for large $\xi$, the secrecy rate bounds increase slowly and then tend to stable values as the increase of $\xi$.
Moreover, with the increase of $P$, the secrecy rate performance also improves.
This indicates that an indoor VLC system with larger nominal optical intensity has better performance.

\begin{figure}
\centering
\includegraphics[width=8.5cm]{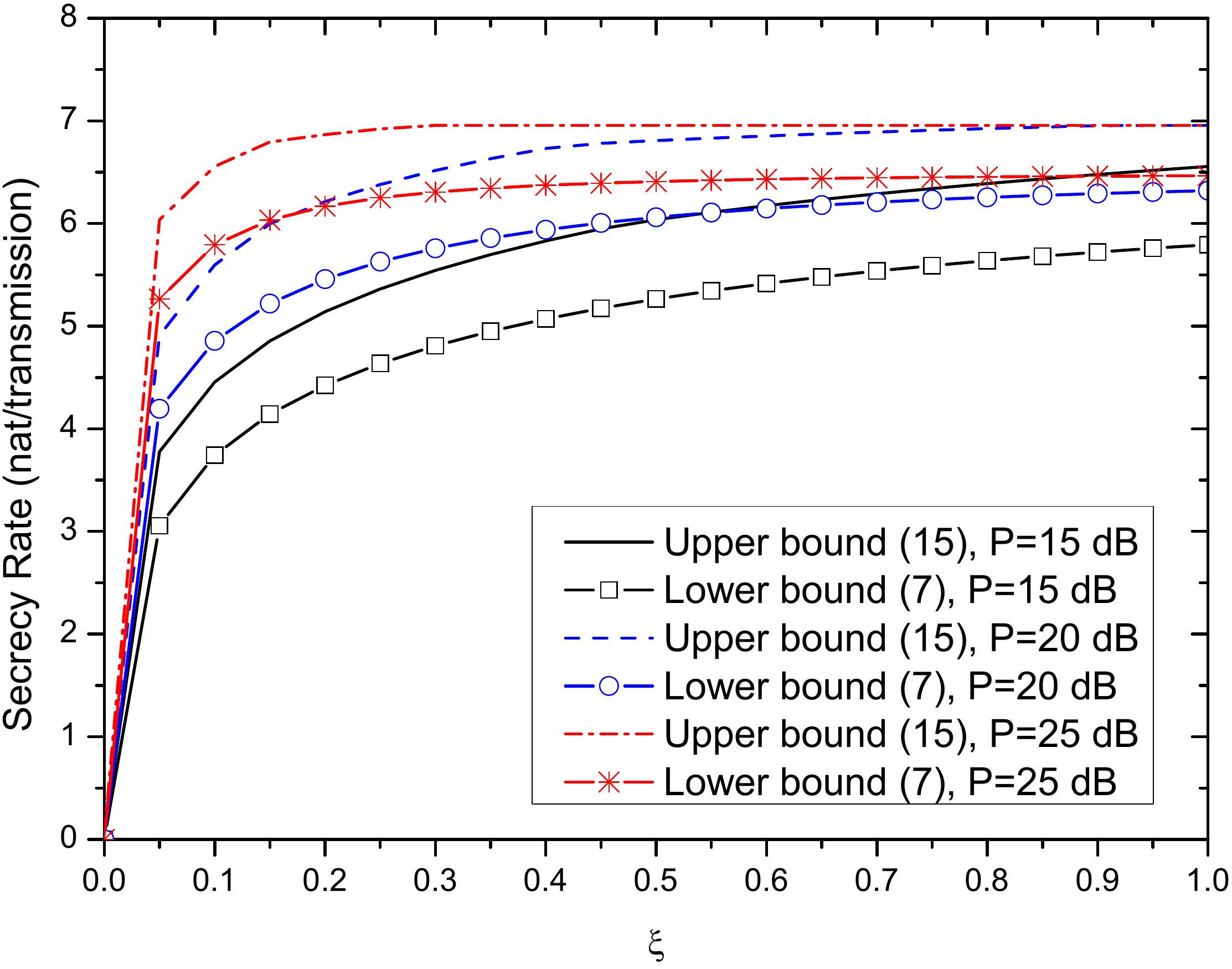}
\caption{Secrecy rate bounds versus $\xi $ with different $P$ when ${h_{\rm{B}}}/{h_{\rm{E}}} = 1000$.}
 \label{fig3}
\end{figure}

Fig. \ref{fig4} plots the secrecy rate bounds versus ${h_{\rm{B}}}/{h_{\rm{E}}}$ with different $P$ when $\xi=0.5$.
As can be observed, when ${h_{\rm{B}}}/{h_{\rm{E}}}<1$,
the main channel is worse than the eavesdropping channel, all secrecy rate bounds are zeros.
In this case, secure information transmission cannot be achieved.
When ${h_{\rm{B}}}/{h_{\rm{E}}} \in (1, 10^{4}]$,
the secrecy rate bounds become nonnegative and increase rapidly with the increase of ${h_{\rm{B}}}/{h_{\rm{E}}}$.
When ${h_{\rm{B}}}/{h_{\rm{E}}} > 10^{4}$, the secrecy rate bounds will not increase anymore.
Moreover, when $h_{\rm{B}}/h_{\rm{E}}$ is small, no matter how large the nominal optical intensity $P$ is, the values of upper bounds on secrecy rate are almost the same as each other. In this case, increasing the transmit optical intensity cannot enhance the secrecy performance.
However, when $h_{\rm{B}}/h_{\rm{E}}$ is large, the trends of secrecy performance improvement become apparent by increasing $P$.

\begin{figure}
\centering
\includegraphics[width=8.5cm]{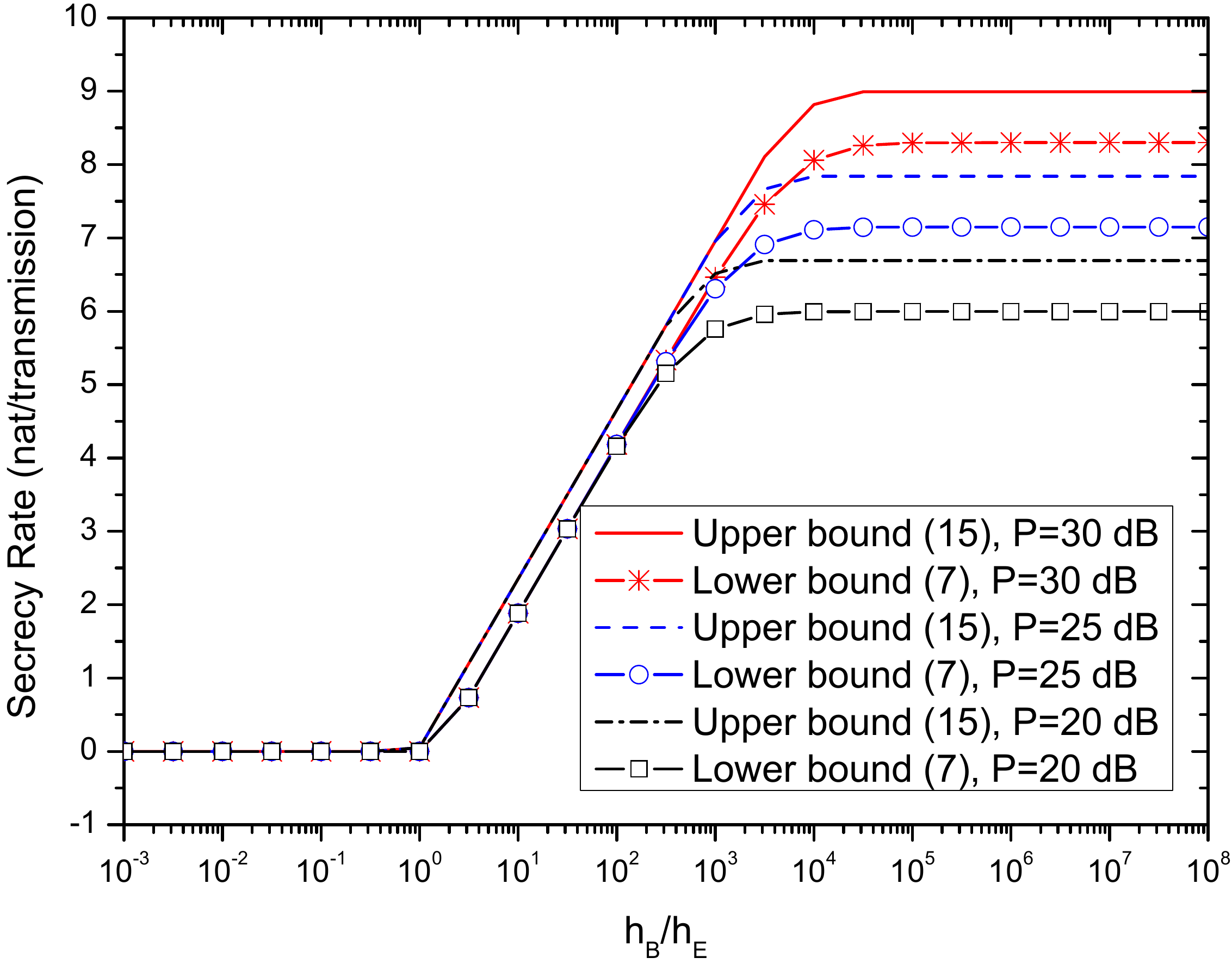}
\caption{Secrecy rate bounds versus ${h_{\rm{B}}}/{h_{\rm{E}}}$ with different $P$ when $\xi=0.5$.}
 \label{fig4}
\end{figure}



\subsection{Results of SM Based VLC with Constraints (\ref{eq3_1}), (\ref{eq3_2}) and (\ref{eq3_3})}

Fig. \ref{fig6} shows the secrecy rate bounds versus $P$ with different ${h_{\rm{B}}}/{h_{\rm{E}}}$ when $A=P$. Specifically, Fig. \ref{fig6}(a) shows the results with $\xi=0.5$, while Fig. \ref{fig6}(b) corresponds to $\xi=0.3$.
Similar to Fig. \ref{fig2}, the secrecy rate bounds first increase and then tend to stable values with the increase of $P$.
Furthermore, as the increase of ${h_{\rm{B}}}/{h_{\rm{E}}}$, the secrecy performance enhances.
When $\xi=0.5$ in Fig. \ref{fig6}(a), the performance gaps between the lower bound (\ref{equ1}) and the upper bound (\ref{equ4}) are small.
To quantitatively quantify the performance gaps when $\xi=0.5$, Table \ref{tab3} is presented.
At high SNR in Table \ref{tab3}, the performance gap between asymptotic lower bound and asymptotic upper bound of secrecy rates is about 0.1765 nat/transmission, which is consistent with the result in \emph{Remark \ref{rem2}}.
When $\xi=0.3$ in Fig. \ref{fig6}(b), the performance gaps between the lower bound (\ref{equ1}) and the upper bound (\ref{equ4}) is provided in Table \ref{tab4}. In this case, the asymptotic performance gap at high SNR is about 0.2676 nat/transmission, which is the same as the result in \emph{Remark \ref{rem3}}.

%

\begin{figure}
\centering
\includegraphics[width=8.5cm]{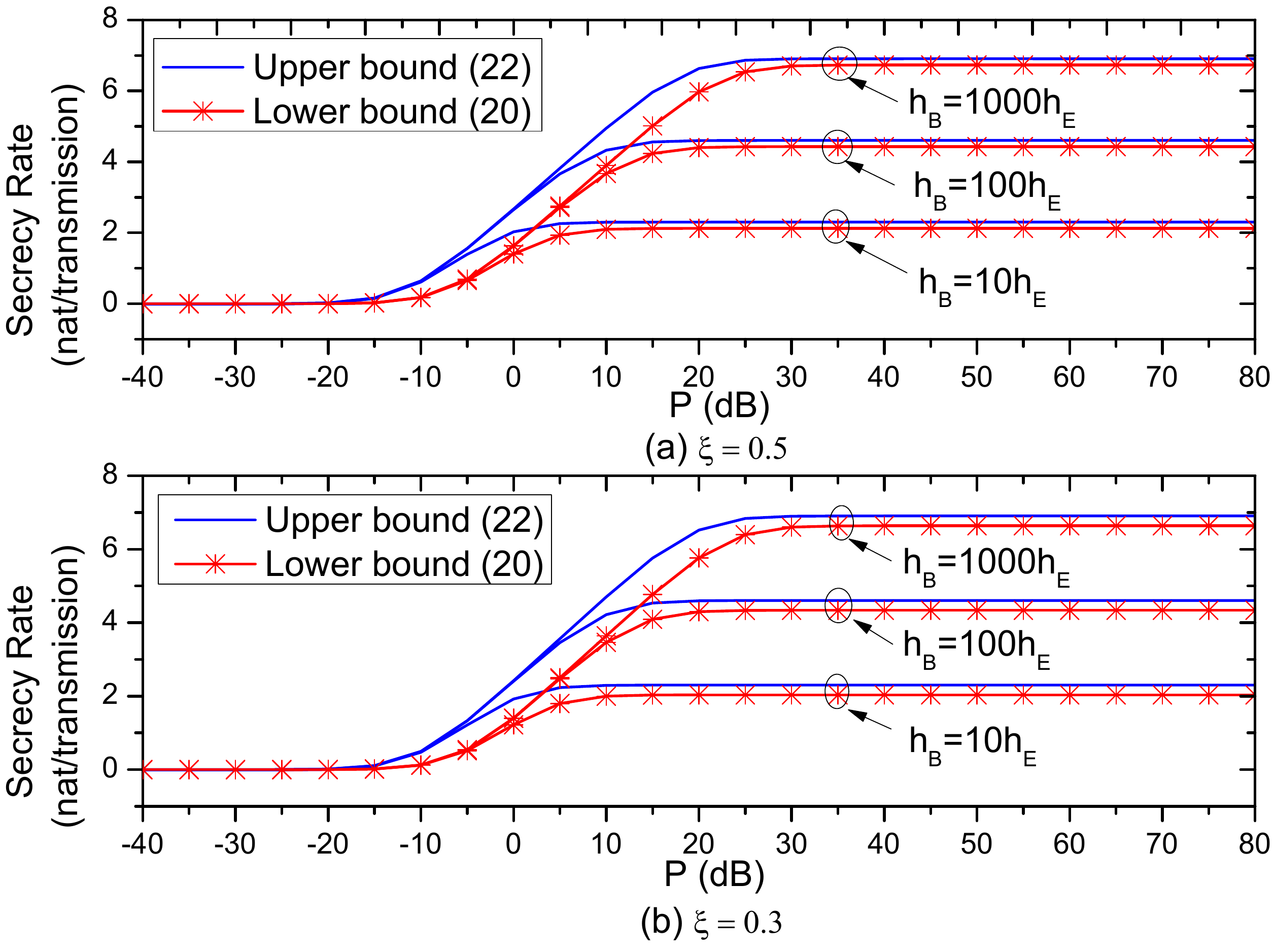}
\caption{Secrecy rate bounds versus $P$ with different $h_{\rm{B}}/h_{\rm{E}}$ when $A=P$.}
 \label{fig6}
\end{figure}

\begin{table}
\caption{Performance gaps between (\ref{equ1}) and (\ref{equ4}) when $A=P$ and $\xi  = 0.5$.}
\centering
\setlength{\tabcolsep}{6pt}
	\begin{tabular}{|c|c|c|c|}
		\hline\hline
        \multicolumn{1}{|c|}{ \multirow{2}*{$P$ ({\rm dB})} }& \multicolumn{3}{c|}{Performance gaps (nat/transmission)} \\
        \cline{2-4}
		{} & ${h_{\rm{B}}} = 10{h_{\rm{E}}}$ & ${h_{\rm{B}}} = 100{h_{\rm{E}}}$ & ${h_{\rm{B}}} = 1000{h_{\rm{E}}}$ \\
        \hline
        30 & 0.17650 & 0.17672 & 0.19849 \\
        \hline
        40 & 0.17649 & 0.17649 & 0.17672 \\
        \hline
        50 & 0.17649 & 0.17649 & 0.17649 \\
        \hline
        60 & 0.17649 & 0.17649 & 0.17649 \\
        \hline
        70 & 0.17649 & 0.17649 & 0.17649 \\
        \hline
        80 & 0.17649 & 0.17649 & 0.17649 \\
        \hline\hline
	\end{tabular}
\label{tab3}	
\end{table}

\begin{table}
\caption{Performance gaps between (\ref{equ1}) and (\ref{equ4}) when $A=P$ and $\xi  = 0.3$.}
\centering
\setlength{\tabcolsep}{6pt}
	\begin{tabular}{|c|c|c|c|}
		\hline\hline
        \multicolumn{1}{|c|}{ \multirow{2}*{$P$ ({\rm dB})} }& \multicolumn{3}{c|}{Performance gaps (nat/transmission)} \\
        \cline{2-4}
		{} & ${h_{\rm{B}}} = 10{h_{\rm{E}}}$ & ${h_{\rm{B}}} = 100{h_{\rm{E}}}$ & ${h_{\rm{B}}} = 1000{h_{\rm{E}}}$ \\
        \hline
        30 & 0.26763 & 0.26792 & 0.75651 \\
        \hline
        40 & 0.26762 & 0.26762 & 0.29605 \\
        \hline
        50 & 0.26762 & 0.26762 & 0.26793 \\
        \hline
        60 & 0.26762 & 0.26762 & 0.26762 \\
        \hline
        70 & 0.26762 & 0.26762 & 0.26762 \\
        \hline
        80 & 0.26762 & 0.26762 & 0.26762 \\
        \hline\hline
	\end{tabular}
\label{tab4}	
\end{table}

Fig. \ref{fig7} plots the secrecy rare bounds versus $\xi$ with different $A$ when $P=A$ and $h_{\rm B}/h_{\rm E}=1000$.
The changing trends of secrecy rate bounds in this figure is different from that in Fig. \ref{fig3}.
As can be seen, the curves of all lower bounds on secrecy rate (\ref{equ1}) are symmetric with respect to $\xi=0.5$,
and the maximum values of the lower bounds are achieved at $\xi=0.5$.
However, the upper bounds of secrecy rate (\ref{equ4}) always increase with the increase of $\xi$.
Moreover, the larger the peak optical intensity $A$ is, the slower the increasing trend of the upper bound (\ref{equ4}) becomes.
Furthermore, with the increase of $A$, the performance gaps between the lower and upper bounds of secrecy rate become smaller and smaller.

\begin{figure}
\centering
\includegraphics[width=8.5cm]{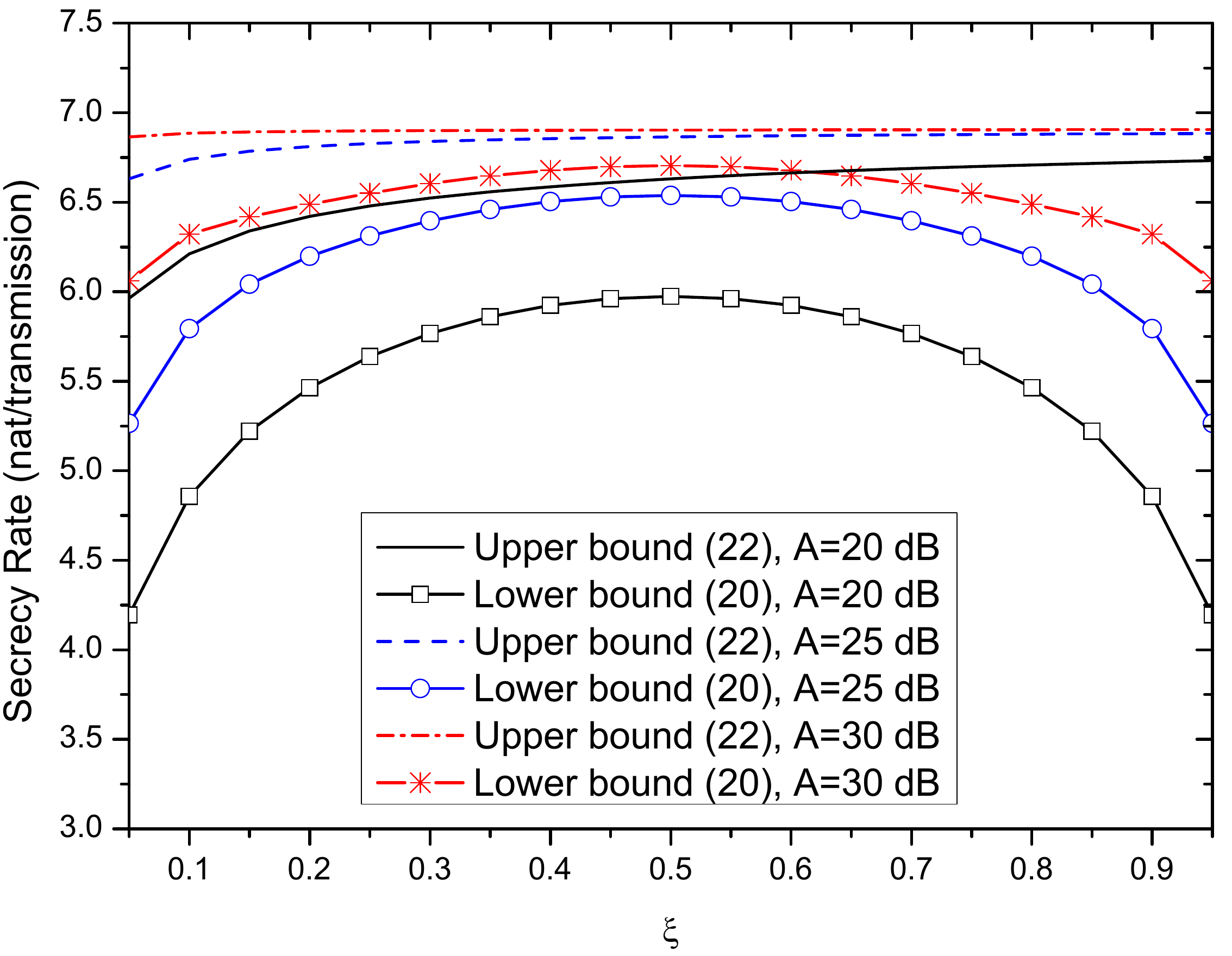}
\caption{Secrecy rare bounds versus $\xi$ with different $A$ when $P=A$ and $h_{\rm B}/h_{\rm E}=1000$.}
 \label{fig7}
\end{figure}

Fig. \ref{fig8} depicts the secrecy rate bounds versus ${h_{\rm{B}}}/{h_{\rm{E}}}$ with different $A$ when $\xi=0.5$ and $A=P$. Similar to Fig. \ref{fig4}, when ${h_{\rm{B}}}/{h_{\rm{E}}}<1$, all secrecy rate bounds are zeros.
The secrecy rate bounds also increase rapidly when ${h_{\rm{B}}}/{h_{\rm{E}}} \in (1, 10^{4}]$, but the secrecy rate bounds do not increase any more when ${h_{\rm{B}}}/{h_{\rm{E}}} >10^{4}$. In Fig. \ref{fig8}, for small ${h_{\rm{B}}}/{h_{\rm{E}}}$, the performance gaps between (\ref{equ1}) and (\ref{equ4}) are small and they can be ignored. Moreover, the secrecy rate bounds do not change with $A$ for small ${h_{\rm{B}}}/{h_{\rm{E}}}$, which is similar to the conclusion in Fig. \ref{fig4}.

\begin{figure}
\centering
\includegraphics[width=8.5cm]{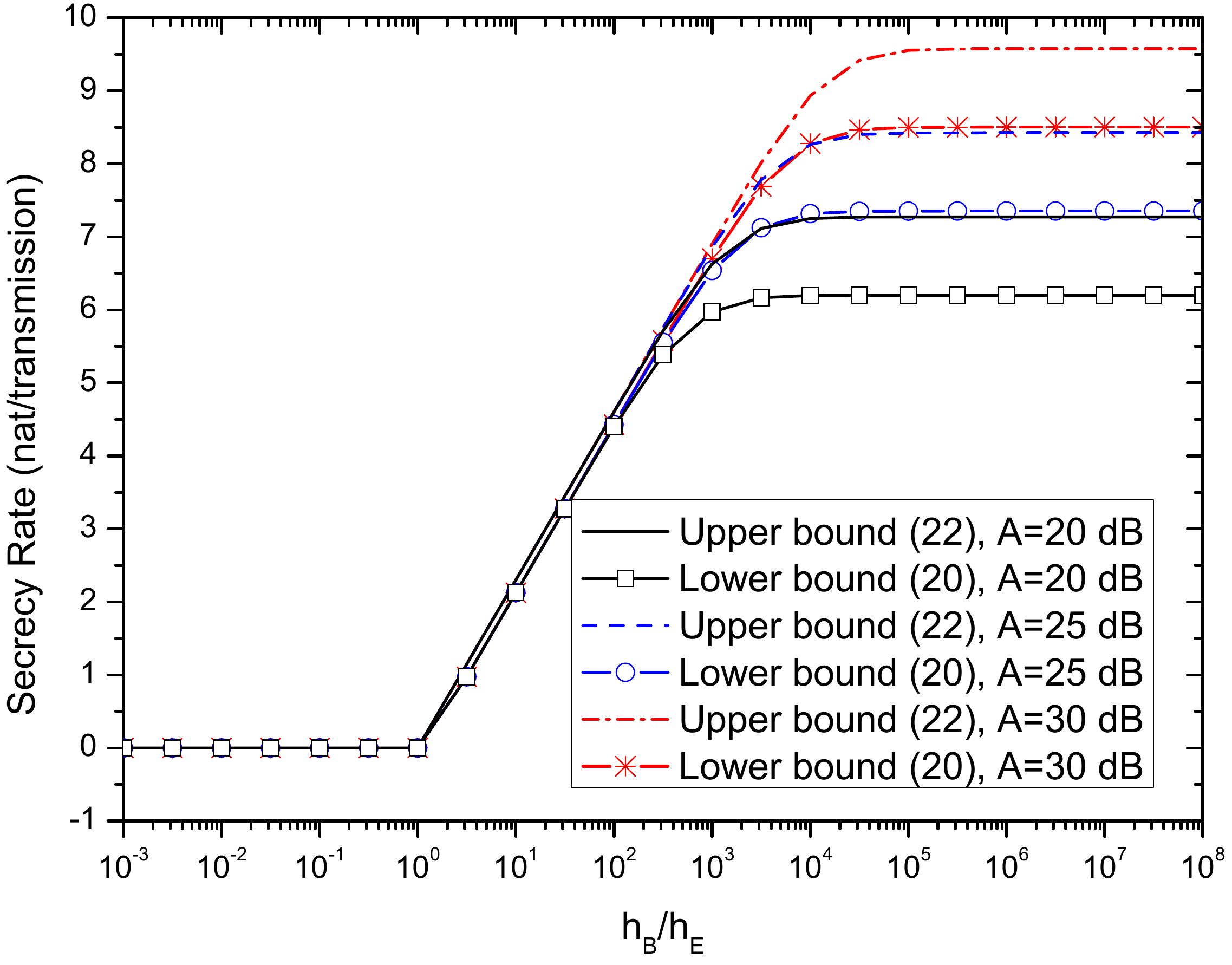}
\caption{Secrecy rate bounds versus ${h_{\rm{B}}}/{h_{\rm{E}}}$ with different $A$ when $\xi=0.5$ and $A=P$.}
 \label{fig8}
\end{figure}

\subsection{Results Comparisons Among the US, CAS, and GS schemes}
When given the positions of Alice, the positions of Bob and Eve have a large impact on the performance of the transmitter selection schemes.
To compare the performance of the US, CAS and GS schemes, the average secrecy rate is used as performance evaluation metric.
Bob traverses his position over the whole receiver plane,
while the position of Eve is selected by satisfying $h_{\rm B}/h_{\rm E}=1000$.

\begin{figure}
\centering
\includegraphics[width=8.5cm]{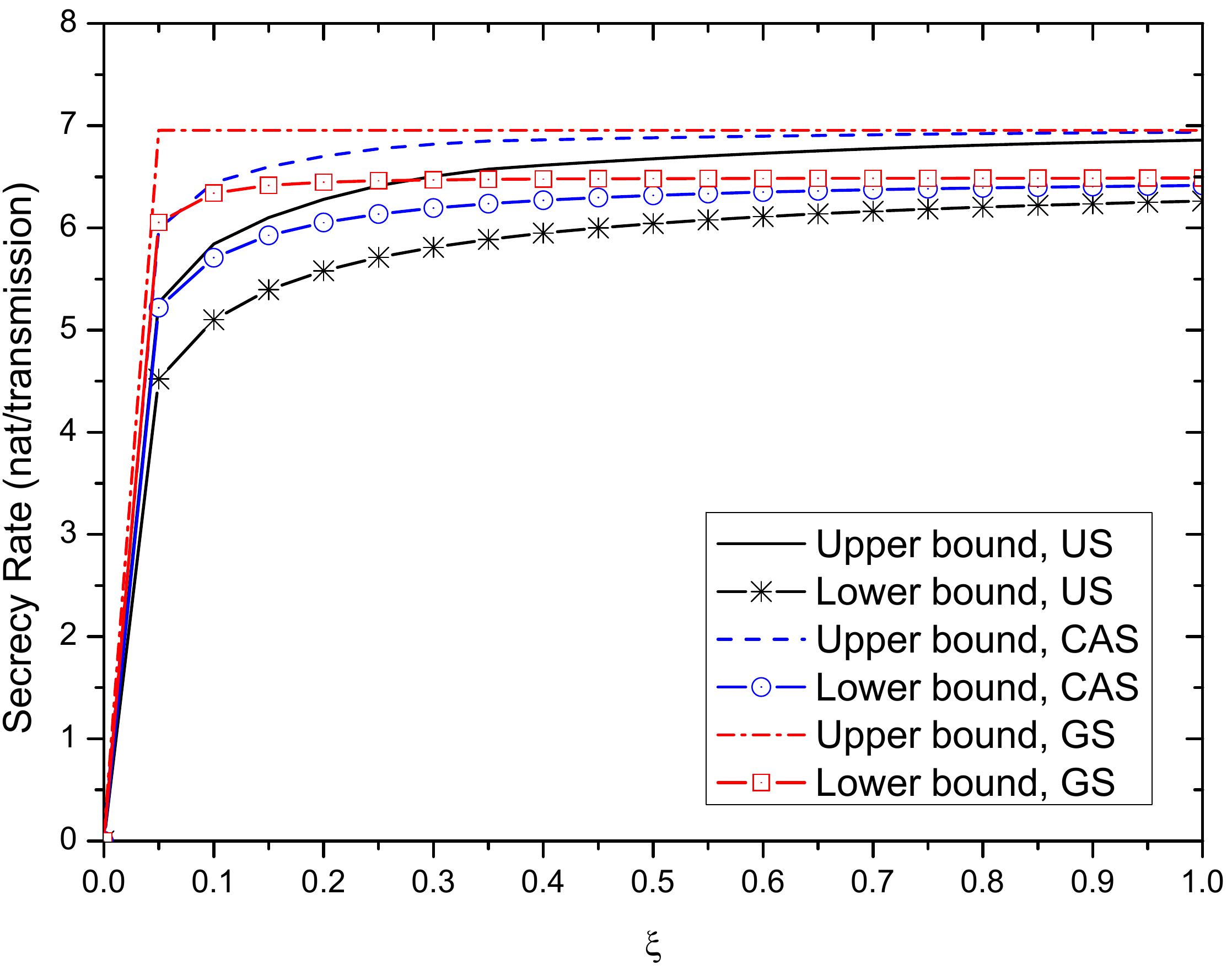}
\caption{Average secrecy rate bounds versus $\xi$ for three transmitter selection schemes when $P=25$ dB and $h_{\rm B}/h_{\rm E}=1000$.}
 \label{fig9}
\end{figure}

\begin{figure}
\centering
\includegraphics[width=8.5cm]{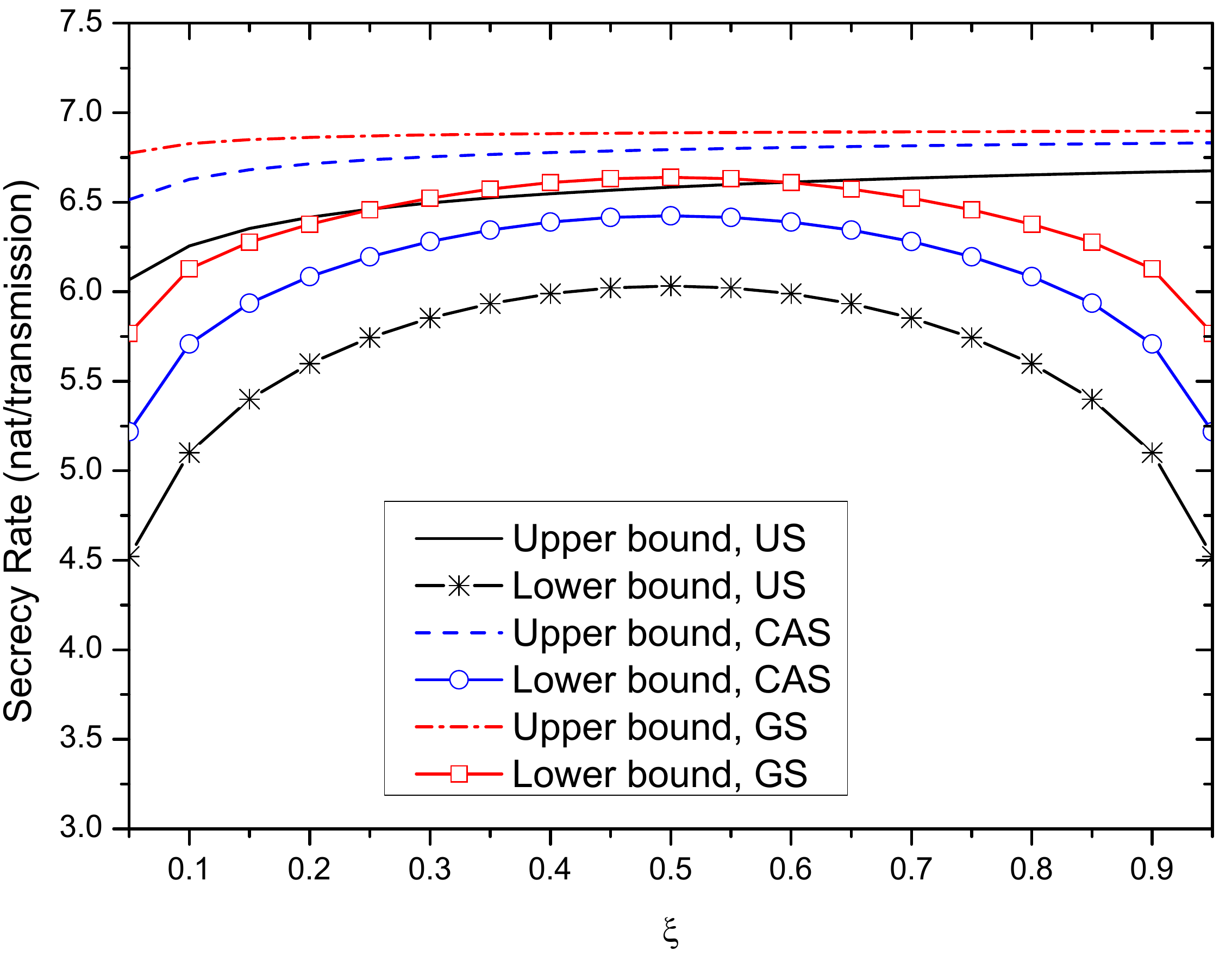}
\caption{Average secrecy rate bounds versus $\xi$ for three transmitter selection schemes when $A=P=25$ dB and $h_{\rm B}/h_{\rm E}=1000$.}
 \label{fig10}
\end{figure}

With the constraints (\ref{eq3_1}) and (\ref{eq3_3}),
Fig. \ref{fig9} shows the average secrecy rate bounds versus $\xi$ for three transmitter selection schemes when $P=25$ dB and $h_{\rm B}/h_{\rm E}=1000$.
As can be seen, the largest secrecy rate is achieved by the GS scheme,
the second largest secrecy rate is obtained by the CAS scheme,
and the smallest secrecy rate is got by the US scheme.
This indicates that selecting transmitter with equal probability may not be a good scheme in practical SM system.
The newly proposed GS and CAS schemes can provide large secrecy performance gains over the US scheme.
When $\xi$ is small, the performance gains of the GS and CAS schemes are dramatically.
With the increase of $\xi$, such performance gains tend to diminish.

With the constraints (\ref{eq3_1}), (\ref{eq3_2}) and (\ref{eq3_3}),
Fig. \ref{fig10} shows the average secrecy rate bounds versus $\xi$ for three transmitter selection schemes when $A=P=25$ dB and $h_{\rm B}/h_{\rm E}=1000$.
Similar to Fig. \ref{fig9}, the GS scheme achieves the largest secrecy rate,
while the US scheme is the worst scheme.
Different from Fig. \ref{fig9}, the secrecy performance gains of the GS and CAS schemes in this figure are obvious for all $\xi$.

\section{Conclusions}
\label{section6}
The secrecy performance for the indoor VLC system using SM scheme is studied in this paper.
The VLC system is consisted of $M$ transmitters, a legitimate receiver and an eavesdropper.
At each time instant, only one transmitter is active via employing the SM scheme.
The US scheme is used to choose the active transmitter.
Under the non-negativity and average optical intensity constraints,
the lower and upper bounds on secrecy rate are derived, respectively.
By considering an additional peak optical intensity constraint,
newly secrecy rate bounds are further obtained.
At high SNR, the asymptotic performance gaps between asymptotic lower and upper bounds are small.
Numerical results verify the tightness of the derived lower and upper bounds.
To further improve the secrecy performance, the CAS and GS schemes are proposed.

In this paper, both Bob and Eve are fixed on the floor.
However, when considering the randomness of the receivers' positions, the derived results in this paper can not be used.
Therefore, analyzing the stochastic secrecy performance is the natural next step.


\numberwithin{equation}{section}
\appendices
\section{Proof of Lower Bound (\ref{eq6}) in Theorem \ref{them1}}
\label{appa}
\renewcommand{\theequation}{A.\arabic{equation}}
The objective function in (\ref{eq5}) can be further rewritten as
\begin{eqnarray}
{R_s} &=& \mathop {\max }\limits_{{f_X}\left( x \right)} \frac 1M \sum\limits_{m = 1}^M \left[  {{\cal I}(X,{h_{{\rm{B}},m}};{Y_{\rm{B}}})} - \sum\limits_{m = 1}^M {{\cal I}(X,{h_{{\rm{E}},m}};{Y_{\rm{E}}})}  \right] \nonumber \\
 &=& \mathop {\max }\limits_{{f_X}\left( x \right)} \frac{1}{M}\sum\limits_{m = 1}^M \left[{\cal H}\left( {{Y_{\rm{B}}}} \right) - {\cal H}\left( {{Y_{\rm{E}}}} \right)+ {\cal H}\left( {{Y_{\rm{E}}}\left| {{h_{{\rm{E}},m}},X} \right.} \right) - {\cal H}\left( {{Y_{\rm{B}}}\left| {{h_{{\rm{B}},m}},X} \right.} \right) \right].
 \label{eq7}
\end{eqnarray}
where ${\cal H}( \cdot )$ denotes the entropy.

Referring to (\ref{eq1}), the PDF of $f_{\left. {{Y_{k}}} \right|{h_{{k},m}},X}\left( {\left. {{y_{k}}} \right|{h_{{k},m}},x} \right)$ ($k=$B or E) can be written as
\begin{eqnarray}
f_{{\left. {{Y_{k}}} \right|{h_{{k},m}},X}}\!\left( {\left. {{y_{k}}} \right|{h_{{k},m}},x} \right) \!\!=\!\! \frac{1}{{\sqrt {2\pi}\sigma _{k} }}e^{ - \frac{{{{({y_{k}} \!-\! {h_{{k},m}}x)}^2}}}{{2\sigma _{k}^2}}}.
\label{eq11}
\end{eqnarray}
Therefore, ${\cal H}\left( {\left. {{Y_{k}}} \right|{h_{{k},m}},X} \right)$ is derived as
\begin{eqnarray}
{\cal H}\left( {\left. {{Y_{k}}} \right|{h_{{k},m}},X} \right) = \frac{1}{2}\ln \left( {2\pi e\sigma _{k}^2} \right).
\label{eq11}
\end{eqnarray}
Substituting (\ref{eq11}) into (\ref{eq7}), we have
\begin{equation}
{R_s} = \mathop {\max }\limits_{{f_X}\left( x \right)} \frac{1}{M}\sum\limits_{m = 1}^M {\left[ {{\cal H}\left( {{Y_{\rm{B}}}} \right) - {\cal H}\left( {{Y_{\rm{E}}}} \right)} \right] + \ln \left( {\frac{{{\sigma _{\rm{E}}}}}{{{\sigma _{\rm{B}}}}}} \right)}.
 \label{eq12}
\end{equation}

According to the entropy-power inequality \cite{BIB16}, ${\cal H}\left( {{Y_{\rm{B}}}} \right)$ in (\ref{eq12}) can be lower-bounded by
\begin{eqnarray}
{\cal H}\left( {{Y_{\rm{B}}}} \right) &=& {\cal H}\left( {{h_{{\rm{B}},m}}X + {Z_{\rm{B}}}} \right) \nonumber \\
 &\ge&
 \frac{1}{2}\ln \left[ {{e^{2\left[ {{\cal H}(X) + \ln ({h_{{\rm{B}},m}})} \right]}} + 2\pi e\sigma _{\rm{B}}^2} \right].
 \label{eq13}
\end{eqnarray}

Moreover, ${\cal H}\left( {{Y_{\rm{E}}}} \right)$ is given by
\begin{equation}
{\cal H}\left( {{Y_{\rm{E}}}} \right) = \frac{1}{2}\ln \left[ {2\pi e{\mathop{\rm var}} \left( {{Y_{\rm{E}}}} \right)} \right].
 \label{eq14}
\end{equation}
where ${\mathop{\rm var}} ( \cdot )$ denotes the variance of a random variable.

Substituting (\ref{eq13}) and (\ref{eq14}) into (\ref{eq12}), the lower bound of the secrecy rate is given by
\begin{eqnarray}
{R_s} \ge \mathop {\max }\limits_{{f_X}\left( x \right)} \frac{1}{{2M}}\sum\limits_{m = 1}^M {\ln \left[ {\frac{{{e^{2\left[ {{\cal H}( X) + \ln ({h_{{\rm{B}},m}})} \right]}} + 2\pi e\sigma _{\rm{B}}^2}}{{2\pi e{\mathop{\rm var}} \left( {{Y_{\rm{E}}}} \right)}}} \right]}
 + \ln \left( {\frac{{{\sigma _{\rm{E}}}}}{{{\sigma _{\rm{B}}}}}} \right).
 \label{eq15}
\end{eqnarray}
To obtain a tight lower bound, a good input PDF can be derived by maximizing ${\cal H}(X)$ under constraints (\ref{eq3_1}) and (\ref{eq3_3}).
Referring to \cite{BIB13}, the optimal input PDF is given by
\begin{equation}
{f_X}\left( x \right) = \frac{1}{{\xi P}}{e^{ - \frac{1}{{\xi P}}x}},x \ge 0.
 \label{eq17}
\end{equation}
Based on the input PDF (\ref{eq17}), we can get
\begin{equation}
\left\{ {\begin{array}{*{20}{c}}
{{\cal H}\left( X \right) = \ln \left( {e\xi P} \right)}\\
{{\mathop{\rm var}} \left( {{Y_{\rm{E}}}} \right) = h_{{\rm{E}},m}^2{\xi ^2}{P^2} + \sigma _{\rm{E}}^2}
\end{array}} \right..
 \label{eq18}
\end{equation}
Substituting (\ref{eq18}) into (\ref{eq15}), \emph{Theorem \ref{them1}} holds.

\section{Proof of (\ref{eq21}) in Lemma \ref{le1}}
\label{appc}
\renewcommand{\theequation}{B.\arabic{equation}}
According to the definitions, we have \cite{BIB19}
\begin{eqnarray}
I(\!X;{Y'_{{\rm B},m}}\!\left| {{Y'_{{\rm E},m}}} \right.\! ) \!\!=\!\!\! \int_0^{ \infty }\!\!\!\! \int_{ - \infty }^{ \infty } \!\!\int_{ - \infty }^{ \infty }\!\!\!\! {f_{X{Y'_{{\rm B},m}}\!{Y'_{{\rm E},m}}}}\!\!(\! {x,\!{y'_{{\rm B},m}},\!{y'_{{\rm E},m}}}\!)\!\! \ln \!\!\frac{{{f_{{Y'_{{\rm B},m}}\!\left| {X{Y'_{{\rm E},m}}} \right.}}\!\!\!\!\!\left( {{y'_{{\rm B},m}}\!\!\left| {x,{y'_{{\rm E},m}}} \right.} \!\right)}}{{{f_{{Y'_{{\rm B},m}}\!\left| {{Y'_{{\rm E},m}}} \right.}}\!\!\!\!\left( {{y'_{{\rm B},m}}\left| {{y'_{{\rm E},m}}} \right.} \right)}}{\rm{d}}{y'_{{\rm B},m}}\!{\rm{d}}{y'_{{\rm E},m}}\!{\rm{d}}x,
\label{eqb1}
\end{eqnarray}
and
\begin{eqnarray}
&&\!\!\!\!\!\!\!\!\!{\mathbb{E}_{X{Y'_{{\rm E},m}}}}\!\!\!\left[\!\! {D\!\!\left(\!\! {{f_{{Y'_{{\rm B},m}}\left| {{Y'_{{\rm E},m}}} \right.}}\!\!\!\!\left( {{y'_{{\rm B},m}}\left| {{Y'_{{\rm E},m}}} \right.} \!\!\right)\!\!\left\| {{g_{{Y'_{{\rm B},m}}\left| {{Y'_{{\rm E},m}}} \right.}}\!\!\!\!\left( {{y'_{{\rm B},m}}\left| {{Y'_{{\rm E},m}}} \right.}\!\! \right)}\!\! \right.} \right)}\!\! \right]\nonumber \\
&&\!\!\!\!\!\!\! =  \!\!{\int_0^{ + \infty }\!\! {\int_{ - \infty }^{ + \infty }\!\! {\int_{ - \infty }^{ + \infty } \!\! {{f_{X{Y'_{{\rm B},m}}{Y'_{{\rm E},m}}}}\left( {x,{y'_{{\rm B},m}},{y'_{{\rm E},m}}} \right)} } } }  \ln \frac{{{f_{{Y'_{{\rm B},m}}\left| {{Y'_{{\rm E},m}}} \right.}}\left( {{y'_{{\rm B},m}}\left| {{y'_{{\rm E},m}}} \right.} \right)}}{{{g_{{Y'_{{\rm B},m}}\left| {{Y'_{{\rm E},m}}} \right.}}\left( {{y'_{{\rm B},m}}\left| {{y'_{{\rm E},m}}} \right.} \right)}}{\rm d}{y'_{{\rm B},m}}{\rm d}{y'_{{\rm E},m}}{\rm d}x.
\label{eqb2}
\end{eqnarray}
Combining (\ref{eqb1}) with (\ref{eqb2}), we can get
\begin{eqnarray}
 {\mathbb{E}_{X{Y'_{{\rm E},m}}}}\!\!\!\left[\!\! {D\!\!\left(\!\! {{f_{{Y'_{{\rm B},m}}\left| {{Y'_{{\rm E},m}}} \right.}}\!\!\!\!\left( {{y'_{{\rm B},m}}\left| {{Y'_{{\rm E},m}}} \right.} \!\!\right)\!\!\left\| {{g_{{Y'_{{\rm B},m}}\left| {{Y'_{{\rm E},m}}} \right.}}\!\!\!\!\left( {{y'_{{\rm B},m}}\left| {{Y'_{{\rm E},m}}} \right.}\!\! \right)}\!\! \right.} \right)}\!\! \right] +I(X;{Y'_{{\rm B},m}}\left| {{Y'_{{\rm E},m}}} \right. )  = {\mathbb{E}_{X{Y'_{{\rm E},m}}}}\{u\}.
 \label{eq20}
\end{eqnarray}
Because the relative entropy on the left hand side of (\ref{eq20}) is non-negative, \emph{Lemma \ref{le1}} holds.

\section{Proof of Lower Bound (\ref{eq25}) in Theorem \ref{them2}}
\label{appb}
\renewcommand{\theequation}{C.\arabic{equation}}
According to the information theory, eq. (\ref{eq24}) can be further written as
\begin{eqnarray}
&&\!\!\!\!\!\!{R_s} \!\!\le\!\! \frac{1}{M} \!\!\sum\limits_{m = 1}^M \!\!\left\{\!\!\underbrace {{{\mathbb{E}}_{{X^*}}}\!\!\left\{\! {\int_{ - \infty }^{ + \infty } \!\!\!\! {\int_{ - \infty }^{ + \infty }\!\!\!\! {{f_{{Y'_{{\rm B},m}}{Y'_{{\rm E},m}}\!\left| {X} \right.}}\!\!\!\left( {{y'_{{\rm B},m}},{y'_{{\rm E},m}}\!\left| {X} \right.}\! \right)\!\!\ln\!\! \left[\! {{f_{{Y'_{{\rm B},m}}\left| {X{Y'_{{\rm E},m}}} \right.}}\!\!\!\left( {{y'_{{\rm B},m}}\left| {X,{y'_{{\rm E},m}}} \right.} \!\right)}\!\! \right]} }\! {\rm{d}}{y'_{{\rm B},m}}{\rm{d}}{y'_{{\rm E},m}}}\!\! \right\}}_{{I_1}} \right. \nonumber \\
&& \left.+\! \underbrace {{{\mathbb{E}}_{{X^*}}}\!\!\!\left\{\!\! { -\!\! \int_{ - \infty }^{ + \infty }\!\!\!\! {\int_{ - \infty }^{ + \infty } \!\!\!\! {{f_{{Y'_{{\rm B},m}}{Y'_{{\rm E},m}}\left| {X} \right.}}\!\!\!\left( {{y'_{{\rm B},m}},{y'_{{\rm E},m}}\!\left| {X} \right.} \right)\!\ln\!\! \left[\! {{g_{{Y'_{{\rm B},m}}\left| {{Y'_{{\rm E},m}}} \right.}}\!\!\!\left( {{y'_{{\rm B},m}}\left| {{y'_{{\rm E},m}}} \right.} \right)} \right]} } {\rm{d}}{y'_{{\rm B},m}}{\rm{d}}{y'_{{\rm E},m}}} \right\}}_{{I_2}}\right\}.
 \label{eq26}
\end{eqnarray}

In (\ref{eq26}), ${I_1}$ can be expressed as
\begin{eqnarray}
{I_1} &=&  - {\cal H}\left( {{Y'_{{\rm B},m}}\left| {{X^*},{Y'_{{\rm E},m}}} \right.} \right) \nonumber \\
 &=&  {\cal H}\!\left( {{Y'_{{\rm E},m}}\left| {{X^*}} \right.} \right) - {\cal H}\!\left( {{Y'_{{\rm B},m}}\left| {{X^*}} \right.} \right) - {\cal H}\left( {{Y'_{{\rm E},m}}\left| {{X^*},{Y'_{{\rm B},m}}} \right.} \right),
  \label{eq27}
\end{eqnarray}
where ${\cal H}( {{{Y'_{k,m}}} |{X^*}})$ ($k=$B or E) can be obtain as
\begin{equation}
{\cal H}\left( {\left. {{Y'_{k,m}}} \right|{X^*}} \right) = {\cal H}\left( {\left. {{Y'_{k,m}}} \right|X} \right) = \frac{1}{2}\ln \left( { \frac{2\pi e \sigma _{k}^2}{h_{k,m}^2}} \right).
\label{eq28}
\end{equation}
Moreover, we have
\begin{eqnarray}
{\cal H}\left( {{Y'_{{\rm E},m}}\left| {{X^*},{Y'_{{\rm B},m}}} \right.} \right) 
=\frac{1}{2}\ln \left[ {2\pi e\!\!\left( \frac{\sigma _{\rm{E}}^2}{h_{{\rm E},m}^2} \!+\! \frac{\sigma _{\rm{B}}^2}{h_{{\rm B},m}^2} \right)}\!\! \right].
\label{eq30}
\end{eqnarray}
Substituting (\ref{eq28}) and (\ref{eq30}) into (\ref{eq27}), we have
\begin{equation}
{I_1} =  - \frac{1}{2}\ln \left[ { \frac{2\pi e\sigma_{\rm B}^2}{h_{{\rm B},m}^2} \left( {1 + \frac{{h_{{\rm{E}},m}^2\sigma _{\rm{B}}^2}}{{h_{{\rm{B}},m}^2\sigma _{\rm{E}}^2}}} \right)} \right].
\label{eq31}
\end{equation}

To obtain ${I_2}$, ${g_{{Y'_{{\rm B},m}}\left| {{Y'_{{\rm E},m}}} \right.}}\left( {{y'_{{\rm B},m}}\left| {{y'_{{\rm E},m}}} \right.} \right)$ is chosen as
\begin{equation}
{g_{{Y'_{{\rm B},m}}\left| {{Y'_{{\rm E},m}}} \right.}}\left( {{y'_{{\rm B},m}}\left| {{y'_{{\rm E},m}}} \right.} \right) = \frac{1}{{2{s^2}}}{e^{ - \frac{{\left| {{y'_{{\rm B},m}} - \mu {y'_{{\rm E},m}}} \right|}}{{{s^2}}}}},
\label{eq32}
\end{equation}
where $\mu $ and $s$ are free parameters to be determined \cite{BIB13}.

Moreover, ${f_{{Y'_{{\rm{B}},m}}{Y'_{{\rm E},m}}|X}}({y'_{{\rm{B}},m}},{y'_{{\rm{E}},m}}|X)$ is given by
\begin{eqnarray}
{f_{{Y'_{{\rm{B}},m}}{Y'_{{\rm E},m}}|X}}({y'_{{\rm{B}},m}},{y'_{{\rm{E}},m}}|X) = \frac{{{e^{ - \frac{{{{({y'_{{\rm{B}},m}} - X)}^2}}}{{2\frac{{\sigma _{\rm{B}}^2}}{{h_{{\rm{B}},m}^2}}}}}}}}{{\sqrt {2\pi \frac{{\sigma _{\rm{B}}^2}}{{h_{{\rm{B}},m}^2}}} }}\frac{{{e^{ - \frac{{{{\left( {{y'_{{\rm{E}},m}} - {y'_{{\rm{B}},m}}} \right)}^2}}}{{2\left( {\frac{{\sigma _{\rm{B}}^2}}{{h_{{\rm{B}},m}^2}} + \frac{{\sigma _{\rm{E}}^2}}{{h_{{\rm{E}},m}^2}}} \right)}}}}}}{{\sqrt {2\pi \left( {\frac{{\sigma _{\rm{B}}^2}}{{h_{{\rm{B}},m}^2}} + \frac{{\sigma _{\rm{E}}^2}}{{h_{{\rm{E}},m}^2}}} \right)} }}.
\label{eq32_1}
\end{eqnarray}
Then, ${I_2}$ can be obtained as
\begin{eqnarray}
{I_2} 
 = \ln (2{s^2}) \!+\! \frac{1}{{{s^2}}}{\mathbb{E}_{{X^*}}}\!\!\left[ {\int_{ - \infty }^\infty \! {\frac{{{e^{ - \frac{{{{({y'_{{\rm{B}},m}} - X)}^2}}}{{2\frac{{\sigma _{\rm{B}}^2}}{{h_{{\rm{B}},m}^{\rm{2}}}}}}}}}}{{\sqrt {2\pi } \frac{{{\sigma _{\rm{B}}}}}{{{h_{{\rm{B}},m}}}}}}\!\!\int_{ - \infty }^\infty \!\! {\frac{{{e^{ - \frac{{{t^2}}}{{2\left( {\frac{{\sigma _{\rm{B}}^2}}{{h_{{\rm{B}},m}^2}} + \frac{{\sigma _{\rm{E}}^2}}{{h_{{\rm{E}},m}^2}}} \right)}}}}}}{{\sqrt {2\pi \!\!\left(\! {\frac{{\sigma _{\rm{B}}^2}}{{h_{{\rm{B}},m}^2}} \!+\! \frac{{\sigma _{\rm{E}}^2}}{{h_{{\rm{E}},m}^2}}} \!\right)} }}\left| {(1 \!-\! \mu ){y'_{{\rm{B}},m}} \!-\! \mu t} \right|{\rm{d}}t{\rm{d}}{y'_{{\rm{B}},m}}} } } \right]\!\!.
\label{eq33}
\end{eqnarray}
Because $\left| {a - b} \right| \le \left| a \right| + \left| b \right|$ and $\left| {a + b} \right| \le \left| a \right| + \left| b \right|$, eq. (\ref{eq33}) can be further upper-bounded by
\begin{eqnarray}
{I_2} \le \ln (2{s^2}) + \frac{2}{{{s^2}}} \underbrace {\left[ {\left| \mu  \right|\sqrt {\frac{{\frac{{\sigma _{\rm{B}}^2}}{{h_{{\rm{B}},m}^2}} \!+\! \frac{{\sigma _{\rm{E}}^2}}{{h_{{\rm{E}},m}^2}}}}{{2\pi }}}  \!+\! \left| {1 \!-\! \mu } \right|\left( {\frac{{{\sigma _{\rm{B}}}}}{{\sqrt {2\pi } {h_{{\rm{B}},m}}}} \!+\! \frac{{\xi P}}{2}} \right)}\!\! \right]}_{{I_3}}\!\!.
\label{eq34}
\end{eqnarray}
To get a relatively tight upper bound of ${I_2}$, the minimum value of ${I_3}$ in (\ref{eq34}) should be determined first.
Here, three cases are considered:

Case 1: when $\mu  < 0$, we have
\begin{equation}
{I_3} \ge \frac{{{\sigma _{\rm{B}}}}}{{\sqrt {2\pi } {h_{{\rm{B}},m}}}} + \frac{{\xi P}}{2}.
\label{eq35}
\end{equation}

Case 2: when $0 \le \mu  \le {\rm{1}}$,
if $\sqrt{\left({\frac{{\sigma _{\rm{B}}^2}}{{h_{{\rm{B}},m}^2}} + \frac{{\sigma _{\rm{E}}^2}}{{h_{{\rm{E}},m}^2}}}\right)/(2\pi)}  \ge \frac{{{\sigma _{\rm{B}}}}}{{\sqrt {2\pi } {h_{{\rm{B}},m}}}} + \frac{{\xi P}}{2}$, we can also get (\ref{eq35}). Otherwise, we have
\begin{equation}
{I_3} \ge \sqrt {\frac{{\frac{{\sigma _{\rm{B}}^2}}{{h_{{\rm{B}},m}^2}} + \frac{{\sigma _{\rm{E}}^2}}{{h_{{\rm{E}},m}^2}}}}{{2\pi }}}.
\label{eq37}
\end{equation}

Case 3: when $\mu  > 1$, we can also easily obtain (\ref{eq37}).

According to the above three cases, we have
\begin{eqnarray}
{I_3} \ge \left\{ \begin{array}{l}
\frac{{{\sigma _{\rm{B}}}}}{{\sqrt {2\pi } {h_{{\rm{B}},m}}}} + \frac{{\xi P}}{2},\;{\rm{if}}\;\sqrt {\frac{{\frac{{\sigma _{\rm{B}}^2}}{{h_{{\rm{B}},m}^2}} + \frac{{\sigma _{\rm{E}}^2}}{{h_{{\rm{E}},m}^2}}}}{{2\pi }}} \ge \frac{{{\sigma _{\rm{B}}}}}{{\sqrt {2\pi } {h_{{\rm{B}},m}}}} + \frac{{\xi P}}{2}\\
\sqrt {\frac{{\frac{{\sigma _{\rm{B}}^2}}{{h_{{\rm{B}},m}^2}} + \frac{{\sigma _{\rm{E}}^2}}{{h_{{\rm{E}},m}^2}}}}{{2\pi }}} ,\;{\rm otherwise}
\end{array} \right.
\label{eq39}
\end{eqnarray}
Submitting (\ref{eq39}) into (\ref{eq34}), ${I_2}$ is further upper-bounded by
\begin{eqnarray}
{I_{\rm{2}}} \le \left\{ \begin{array}{l}
\ln (2{s^2}) + \frac{2}{{{s^2}}}\left( {\frac{{{\sigma _{\rm{B}}}}}{{\sqrt {2\pi } {h_{{\rm{B}},m}}}} + \frac{{\xi P}}{2}} \right),\;{\rm{if}}\;\sqrt {\frac{{\frac{{\sigma _{\rm{B}}^2}}{{h_{{\rm{B}},m}^2}} + \frac{{\sigma _{\rm{E}}^2}}{{h_{{\rm{E}},m}^2}}}}{{2\pi }}}  \ge \frac{{{\sigma _{\rm{B}}}}}{{\sqrt {2\pi } {h_{{\rm{B}},m}}}} + \frac{{\xi P}}{2}\\
\ln (2{s^2}) + \frac{2}{{{s^2}}}\sqrt {\frac{{\frac{{\sigma _{\rm{B}}^2}}{{h_{{\rm{B}},m}^2}} + \frac{{\sigma _{\rm{E}}^2}}{{h_{{\rm{E}},m}^2}}}}{{2\pi }}} ,\;{\rm otherwise}
\end{array} \right.
\label{eq40}
\end{eqnarray}
To minimize the terms on the right hand side of (\ref{eq40}),
we choose $s^2$ as
\begin{eqnarray}
{s^2} = \left\{ \begin{array}{l}
2\left( {\frac{{{\sigma _{\rm{B}}}}}{{\sqrt {2\pi } {h_{{\rm{B}},m}}}} + \frac{{\xi P}}{2}} \right),{\rm{if}}\sqrt {\frac{{\frac{{\sigma _{\rm{B}}^2}}{{h_{{\rm{B}},m}^2}} + \frac{{\sigma _{\rm{E}}^2}}{{h_{{\rm{E}},m}^2}}}}{{2\pi }}}  \ge \frac{{{\sigma _{\rm{B}}}}}{{\sqrt {2\pi } {h_{{\rm{B}},m}}}} + \frac{{\xi P}}{2}\\
2\sqrt {\frac{{\frac{{\sigma _{\rm{B}}^2}}{{h_{{\rm{B}},m}^2}} + \frac{{\sigma _{\rm{E}}^2}}{{h_{{\rm{E}},m}^2}}}}{{2\pi }}} ,\;{\rm{otherwise}}
\end{array} \right.
\label{eq40_0}
\end{eqnarray}
Submitting (\ref{eq40_0}) to (\ref{eq40}), $I_2$ is finally upper-bounded by
\begin{eqnarray}
{I_{\rm{2}}} \le \left\{ \begin{array}{l}
\ln \left[ {4e\left( {\frac{{{\sigma _{\rm{B}}}}}{{\sqrt {2\pi } {h_{{\rm{B}},m}}}} + \frac{{\xi P}}{2}} \right)} \right],{\rm{if}}\sqrt {\frac{{\frac{{\sigma _{\rm{B}}^2}}{{h_{{\rm{B}},m}^2}} + \frac{{\sigma _{\rm{E}}^2}}{{h_{{\rm{E}},m}^2}}}}{{2\pi }}}  \ge \frac{{{\sigma _{\rm{B}}}}}{{\sqrt {2\pi } {h_{{\rm{B}},m}}}} + \frac{{\xi P}}{2}\\
\ln \left( {4e\sqrt {\frac{{\frac{{\sigma _{\rm{B}}^2}}{{h_{{\rm{B}},m}^2}} + \frac{{\sigma _{\rm{E}}^2}}{{h_{{\rm{E}},m}^2}}}}{{2\pi }}} } \right),\;{\rm{otherwise}}
\end{array} \right.
\label{eq40_1}
\end{eqnarray}
Finally, submitting (\ref{eq31}) and (\ref{eq40_1}) into (\ref{eq26}), eq. (\ref{eq25}) can be derived.

\section{Proof of Lower Bound (\ref{equ1}) in Theorem \ref{them3}}
\label{appd}
\renewcommand{\theequation}{D.\arabic{equation}}
For this scenario, eq. (\ref{eq15}) can also be obtained.
To obtain a good input PDF, the input entropy ${\cal H}(X)$ should be maximized under constraints (\ref{eq3_1}), (\ref{eq3_2}), and (\ref{eq3_3}).

According to \cite{BIB13}, when $\alpha =0.5$, the optimal input PDF is given by
\begin{eqnarray}
{f_X}(x) = \left\{ \begin{array}{l}
\frac{1}{A},\;x \in [0,A]\\
0,\;{\rm{otherwise}}
\end{array} \right..
\label{eqd2}
\end{eqnarray}
Therefore, ${\cal H}(X)$ and ${\rm var}(Y_{\rm E})$ can be written as
\begin{eqnarray}
\left\{ \begin{array}{l}
{\cal H}(X) = \ln A\\
{\mathop{\rm var}} ({Y_{\rm E}}) = h_{{\rm E},m}^2\frac{{{\xi ^2}{P^2}}}{3} + \sigma _{\rm{E}}^2
\end{array} \right..
\label{eqd3}
\end{eqnarray}
Submitting (\ref{eqd3}) into (\ref{eq15}), lower bound (\ref{equ1}) for $\alpha=0.5$ is obtained.

When $\alpha \neq 0.5$ and $\alpha \in (0,1]$, the optimal input PDF is given by \cite{BIB13}
\begin{eqnarray}
{f_X}(x) = \left\{ \begin{array}{l}
\frac{{c{e^{cx}}}}{{{e^{cA}} - 1}},\;x \in [0,A]\\
0,\;{\rm{otherwise}}
\end{array} \right.,
\label{eqd5}
\end{eqnarray}
where $c$ is the solution to (\ref{equ2}).
Therefore, ${\cal H}(X)$ and ${\rm var}(Y_{\rm E})$ can be written as
\begin{eqnarray}
\left\{ \begin{array}{l}
{\cal H}(X) = \ln \left[ {{e^{ - c\xi P}}\left( {\frac{{{e^{cA}} - 1}}{c}} \right)} \right]\\
{\mathop{\rm var}} ({Y_{\rm{E}}}) = h_{{\rm{E}},m}^2\left[ {\frac{{A(cA - 2)}}{{c(1 - {e^{ - cA}})}} + \frac{2}{{{c^2}}} - {\xi ^2}{P^2}} \right] + \sigma _{\rm{E}}^2
\end{array} \right.\!\!\!.
\label{eqd6}
\end{eqnarray}
Submitting (\ref{eqd6}) into (\ref{eq15}), lower bound (\ref{equ1}) for $\alpha \neq 0.5$ and $\alpha \in (0,1]$ is obtained.

\section{Proof of Lower Bound (\ref{equ4}) in Theorem \ref{them4}}
\label{appe}
\renewcommand{\theequation}{E.\arabic{equation}}
In this scenario, eq. (\ref{eq26}) can also be derived, where $I_1$ can also be expressed as (\ref{eq31}).
To derive $I_2$, ${g_{{Y'_{{\rm B},m}}\left| {{Y'_{{\rm E},m}}} \right.}}\left( {{y'_{{\rm B},m}}\left| {{y'_{{\rm E},m}}} \right.} \right)$ is chosen as
\begin{eqnarray}
{g_{{Y'_{{\rm B},m}}\left| {{Y'_{{\rm E},m}}} \right.}}\left( {{y'_{{\rm B},m}}\left| {{y'_{{\rm E},m}}} \right.} \right)=\frac{1}{\sqrt{2 \pi} s}e^{-\frac{(y'_{{\rm B},m}-\mu y'_{{\rm E},m})^2}{2s^2}},
\label{eqe1}
\end{eqnarray}
where $\mu$ and $s$ are free parameters to be determined.

By using (\ref{eqe1}) and (\ref{eq32_1}), we have
\begin{eqnarray}
I_2 = \frac{1}{2} \ln(2 \pi s^2) +\mathbb{E}_{X^*} \left\{ \frac{(1-\mu \frac{h_{{\rm E},m}}{h_{{\rm B},m}})^2 (h_{{\rm B},m}^2 X^2+ \sigma_{\rm B}^2)+\mu^2 (\frac{h_{{\rm E},m}^2}{h_{{\rm B},m}^2} \sigma_{\rm B}^2+\sigma_{{\rm E},m}^2\!)}{2s^2} \right\}.
\label{eqe2}
\end{eqnarray}

According to (\ref{eq3_1}), (\ref{eq3_2}) and (\ref{eq3_3}), we have
\begin{eqnarray}
\mathbb{E}_{X^*}(X^2)=\int_0^A x^2 f_{X^*}(x){\rm d}x \le \int_0^A \!Ax f_{X^*}(x){\rm d}x=A \xi P.
\label{eqe3}
\end{eqnarray}
Then, eq. (\ref{eqe2}) is upper bounded by
\begin{eqnarray}
 I_2 \leq \frac{1}{2} \ln(2 \pi s^2) +  \frac{\left(1\!-\!\mu \frac{h_{{\rm E},m}}{h_{{\rm B},m}}\right)^2 (h_{{\rm B},m}^2 A \xi P\!\!+\!\! \sigma_{\rm B}^2)\!\!+\!\!\mu^2 \left(\frac{h_{{\rm E},m}^2}{h_{{\rm B},m}^2} \sigma_{\rm B}^2\!\!+\!\!\sigma_{{\rm E},m}^2\!\right)}{2s^2}.
\label{eqe4}
\end{eqnarray}
To obtain a tight upper bound of $I_2$, $\mu$ and $s$ are chosen as
\begin{eqnarray}
\left\{ \begin{array}{l}
\mu {\rm{ = }}\frac{{\frac{{{h_{{\rm{E}},m}}}}{{{h_{{\rm{B}},m}}}}(h_{{\rm{B}},m}^2A\xi P + \sigma _{\rm{B}}^{\rm{2}})}}{{h_{{\rm{E}},m}^2A\xi P + 2\frac{{h_{{\rm{E}},m}^2}}{{h_{{\rm{B}},m}^2}}\sigma _{\rm{B}}^{\rm{2}}{\rm{ + }}\sigma _{\rm{E}}^{\rm{2}}}}\\
{s^2} = \frac{{\left( {\frac{{h_{{\rm{E}},m}^2}}{{h_{{\rm{B}},m}^2}}\sigma _{\rm{B}}^{\rm{2}}{\rm{ + }}\sigma _{\rm{E}}^{\rm{2}}} \right)(h_{{\rm{B}},m}^2A\xi P + \sigma _{\rm{B}}^{\rm{2}})}}{{h_{{\rm{E}},m}^2A\xi P + 2\frac{{h_{{\rm{E}},m}^2}}{{h_{{\rm{B}},m}^2}}\sigma _{\rm{B}}^{\rm{2}}{\rm{ + }}\sigma _{\rm{E}}^{\rm{2}}}}
\end{array} \right..
\label{eqe5}
\end{eqnarray}
Submitting (\ref{eqe5}) into (\ref{eqe4}), we have
\begin{eqnarray}
I_2 \leq  \frac{1}{2} \ln \left[2 \pi e \frac{\left(\frac{h_{{\rm E},m}^2}{h_{{\rm B},m}^2} \sigma_{\rm B}^2+ \sigma_{\rm E}^2\right) (h_{{\rm B},m}^2 A\xi P +\sigma_{\rm B}^2)}{h_{{\rm E},m}^2 A \xi P +2 \frac{h_{{\rm E},m}^2}{h_{{\rm B},m}^2}\sigma_{\rm B}^2+\sigma_{\rm E}^2}\right].
\label{eqe6}
\end{eqnarray}
Substituting (\ref{eq31}) and (\ref{eqe6}) into (\ref{eq26}), eq. (\ref{equ4}) can be derived.

\end{document}